%% file: main.tex
\documentclass[reqno]{amsart}
\usepackage[utf8]{inputenc}
\usepackage{mathrsfs}  
\usepackage{proof}
\usepackage[all]{xy}
\usepackage{enumitem}
\usepackage{amsmath}
\usepackage{amssymb}
\usepackage{hyperref}
\hypersetup{%
  colorlinks=true,
  urlbordercolor={1 1 1},
  linkcolor={black},
  raiselinks=false
}

\title[From P-tV to B-eS for IPL]
{From Proof-theoretic Validity \\[1mm] to Base-extension Semantics for \\[1mm] Intuitionistic Propositional Logic}

\author{Alexander V. Gheorghiu}

 \address{Department of Computer Science\\
 University College London\\
 London WC1E 6BT, UK}
 \email{alexander.gheorghiu.19@ucl.ac.uk}

\author{David J. Pym}
 \address{Institute of Philosophy,
 University of London, 
 London WC1E 7HU, UK \and \newline Department of Computer Science and Department of Philosophy,
 University College London,
 London WC1E 6BT, UK}
 \email{d.pym@ucl.ac.uk}

\newtheorem{Theorem}{Theorem}
\newtheorem{Proposition}[Theorem]{Proposition}
\newtheorem{Lemma}[Theorem]{Lemma}
\newtheorem{Corollary}[Theorem]{Corollary}
\theoremstyle{definition}
\newtheorem{Definition}[Theorem]{Definition}

\newtheorem{Example}[Theorem]{Example}
\newtheorem{Condition}[Theorem]{Condition}
\input{macro}
\newcommand{\figureline}{\noindent\rule{\linewidth}{0.4pt}}

\begin{document}

\begin{abstract}
    Proof-theoretic semantics (P-tS) is the approach to meaning in logic based on \emph{proof} (as opposed to truth). There are two major approaches to P-tS: proof-theor\-etic validity (P-tV) and base-extension semantics (B-eS). The former is a semantics of arguments, and the latter is a semantics of logical constants.  This paper demonstrates that the B-eS for \emph{intuitionistic propositional logic} (IPL) encapsulates the declarative content of a  version of P-tV based on the elimination rules. This explicates how the B-eS for IPL works, and shows the completeness of this version of P-tV.
\end{abstract}

\maketitle

\section{Introduction} \label{sec:intro} 

One intuition regarding the meaning of logical consequence, $\Gamma \vdash \phi$, is that it holds by virtue of the logical form of $\Gamma$ and $\phi$, rather than their specific content. One way to express this is by considering arbitrary interpretations of the specific content and demonstrating that $\phi$ holds in any situation in which $\Gamma$ holds. This leads to Tarski's interpretation of consequence based on models $\mathfrak{M}$,
\[
\Gamma \vdash \phi \qquad \mbox{iff} \qquad \mbox{for any model $\mathfrak{M}$, if $\mathfrak{M} \vDash \Gamma$, then $\mathfrak{M} \vDash \phi$}
\]
which defines \emph{model-theoretic semantics} (M-tS). Observe that consequence is defined in terms of the transmission of some categorical notion (in this case, truth). Schroeder-Heister \cite{Schroeder2012categorical} has called this the `standard dogma' of semantics. 

As Prawitz~\cite{Prawitz2007} explains, M-tS conflates the meaning of the logical constants with the \emph{meaning of truth}, since logical structure is defined in terms of interpretations. For example, if $T$ is defined as the least set satisfying certain properties, including ‘$\phi \land \psi \in T$ iff $\phi \in T$ and $\psi \in T$’, then no information is gained about $\land$ by saying that it satisfies this relationship. Moreover, M-tS fails to provide a genuinely consequential relationship between $\Gamma$ and $\phi$.

Tennant~\cite{Tennant78entailment} observes that a consequential reading of a consequence judgment $\Gamma \proves \phi$ implies that $\phi$ follows from $\Gamma$ by some valid reasoning. This requires a notion of a \emph{valid argument} that encapsulates the forms of valid reasoning. We must, therefore, explicate the semantic conditions required for an argument that demonstrates
\[
\mbox{$\psi_1, \ldots, \psi_n$; therefore, $\phi$}
\]
to be valid. Following Prawitz~\cite{Prawitz2007}, these semantic conditions ought to be based on the logical structure of $\psi_1, \ldots, \psi_n$ and some fixed laws of thought. 

Consequently, we abandon the denotationalist perspective on logic, on which M-tS rests, where meaning is given relative to interpretation. Instead, we adopt an \emph{inferentialist} perspective, where meaning is given in terms of inferential relationships — see Brandom~\cite{Brandom2000} and Murzi and Steinberger~\cite{Brandom2000}.

In this paper, we work entirely in the setting of \emph{natural deduction} in the sense of Gentzen~\cite{Gentzen}. In inferentialism, even atomic propositions gain meaning through their \emph{inferential roles}. Thus, we use \emph{atomic systems} to define when the atomic propositions hold (as opposed to using models) details are provided in Section~\ref{sec:background}. Heuristically,  atomic systems are sets of natural deduction rules restricted to atomic propositions. This embodies a ‘meaning-as-use’ philosophy. For example, the proposition `Tammy is a vixen' means ‘Tammy is female’ \emph{and} ‘Tammy is a fox’, governed by these rules:
\[
{ 
\begin{array}{c}
\infer{\text{Tammy is  vixen}}{\text{Tammy is a fox} & \text{Tammy is female}} \qquad
\infer{\text{Tammy is female}}{\text{Tammy is a vixen}} \qquad \infer{\text{Tammy is a fox}}{\text{Tammy is a vixen}} \\
\end{array}
}
\]
These rules, from the inferentialist perspective, are understood as supplying the meaning of the proposition. The `and' above is justified by  comparison with the laws governing conjunction $(\land)$ in $\calculus{NJ}$, 
\[
{
\infer{\phi \land \psi}{\phi & \psi} \qquad
\infer{\phi}{\phi \land \psi} \qquad \infer{\psi}{\phi \land \psi} 
}
\] 
There are important philosophical and mathematical ramifications on the structure of atomic system admitted --- see, for example, Sandqvist~\cite{Sandqvist2009CL,Sandqvist2015IL} and Piecha and Schroeder-Heister ~\cite{Schroeder2016atomic,Piecha2017definitional}. 

The area of logic concerning such a consequentialist reading of logic is \emph{proof-theoretic semantics} (P-tS)~\cite{SEP-PtS,francez2015proof,wansing2000idea}. It is the area of semantics concerning \emph{proof} (as opposed to truth), where `proof' means `valid argument' (as opposed to derivation in a fixed system).  This includes both semantics \emph{of} proofs (i.e., validity conditions on `arguments') and semantics \emph{in terms of} proofs (i.e., the meaning of logical constants in terms of consequential relationships). We call the first \emph{proof-theoretic validity} (P-tV) and the second \emph{base-extension semantics} (Be-S). This nomenclature follows from certain traditions in the literature, but both branches concern \emph{validity} and make use of \emph{base-extensions} in doing so.

Details of B-eS pertinent to this paper are provided in Section~\ref{sec:bes}. Heuristically, a B-eS is defined by a support judgment $\supp$, relative to atomic systems $\base{B}$, by clauses for logical constants, with the base case given by derivability --- that is, if $p$ is an atomic proposition,
\[
\supp_{\base{B}} p \qquad \mbox{iff} \qquad \proves_{\base{B}} p 
\]
(where $\proves_{\base{B}} p$ indicates $p$ can be proved from the rules in $\base{B}$). This mirrors satisfaction in M-tS but can differ significantly. In particular, taking the standard clause for disjunction
\[
\supp_{\base{B}} \phi \lor \psi \qquad \mbox{iff} \qquad \supp_{\base{B}} \phi \mbox{ or } \supp_{\base{B}}  \psi 
\]
renders IPL incomplete (see Piecha et al.~\cite{Piecha2015failure,Piecha2016completeness,Piecha2019incompleteness}), unless additional structure is added elsewhere (see, for example, Stafford and Nascimento~\cite{stafford2023,nascimentothesis}). Sandqvist~\cite{Sandqvist2015IL} showed that IPL is sound and complete for a notion of support with an alternative clause,
\[
\supp_{\base{B}} \phi \lor \psi \qquad \mbox{iff} \qquad  \mbox{$\forall \base{C} \supseteq \baseB$ and $\forall \at{p} \in \setAtoms$, if $\phi \supp_{\base{C}} \at{p}$ and  $\psi \supp_{\base{C}} \at{p}$, then $\supp_{\base{C}} \at{p}$}
\]
This paper gives an operational account of this clause in the sense that it explains what it says about arguments for $\phi \lor \psi$.

Given a notion of P-tV, consequence is defined as follows:
\[
\text{$\Gamma \vdash \phi$ \qquad iff \qquad there is a valid argument from $\Gamma$ to $\phi$}
\]
Prawitz~\cite{Prawitz1971ideas} conjectured that the original and most widely studied account of P-tV corresponds to IPL, but this remains an open problem. The aforementioned work by Piecha et al.~\cite{Piecha2015failure,Piecha2016completeness,Piecha2019incompleteness} says that the conjecture fails (i.e., IPL is incomplete with respect to the semantics) when the notion of P-tV is slightly simplified. Stafford~\cite{Stafford2021} has shown that the semantics for Piecha et al~\cite{Piecha2015failure,Piecha2016completeness,Piecha2019incompleteness} corresponds to \emph{general inquisitive logic} --- that is, the intermediate logic(s) that extends IPL with the rule
\[
\infer{(H \to \phi) \lor (H \to \psi)}{H \to (\phi \lor \psi)}
\]
where $H$ is a hereditary Harrop formula (see Miller~\cite{miller1989logical}). While this rule is admissible in IPL, it is not derivable --- see Harrop~\cite{harrop1960concerning}.

In this paper, we consider P-tV in the Dummett-Prawitz tradition.  A key motivation lies in the following remarks by Gentzen~\cite{Gentzen}:
\begin{quote}
    The introductions represent, as it were, the `definitions' of the symbols concerned, and the eliminations are no more, in the final analysis, than the consequences of these definitions. This fact may be expressed as follows: In eliminating a symbol, we may use the formula with whose terminal symbol we are dealing only `in the sense afforded it by the introduction of that symbol'
\end{quote}
Prawitz~\cite{Prawitz1971ideas,Prawitz1973towards,Prawitz1974} used his normalization theory for $\calculus{NJ}$ to develop a semantic concept reflecting this intuition. Dummett~\cite{Dummett1991logical} later developed the philosophical underpinnings of the idea. 

The basic idea of P-tV in the Dummett-Prawitz tradition is that arguments are valid by virtue of their form. One begins with some class of canonical proofs relative to which validity is computed. Arguments are valid if they \emph{represent} a canonical proof. Thus, P-tV in the Dummett-Prawitz tradition is based on the following ideas:
\begin{itemize}[label=--]
    \item the priority of canonical proofs
    \item the reduction of closed non-canonical arguments to canonical ones.
    \item the substitutional view of open arguments --- that is,  open arguments are justified by considering their closed instances.
\end{itemize}
We defer to Schroeder-Heister~\cite{Schroeder2006validity} for a formal account of this version of P-tS and its subsequent developments --- see, for example, Prawitz~\cite{Prawitz1972,Prawitz1973towards,Prawitz1974}. This is closely related to the \emph{Brouwer-Heyting-Kolmogorov} (BHK) interpretation of intuitionism --- see Section~\ref{sec:bhk} and Schroeder-Heister~\cite{Schroeder2007modelvsproof}.

Typically, normalized closed arguments are \emph{valid} iff their immediate sub-proofs are valid, prioritizing introduction rules. Normalized derivations in $\calculus{NJ}$ conclude with introduction rules (see Prawitz~\cite{Prawitz2006natural}). Schroeder-Heister~\cite{Schroeder2015elim} proposed an alternative based on elimination rules, drawing ideas by Prawitz~\cite{Prawitz1971ideas}. The logical form of a proposition tells us how we may use it; for example, given an implications proposition, its logical form says no more than this: one may establish the consequent by establishing the antecedent. This is expressed by the law of \emph{modus ponens},
\[
\infer{\phi}{\phi \to \psi & \phi}
\]
More generally, it is the elimination (not introduction) rules that says how one may use a proposition of a certain logical form. This suggests a version of P-tV based on elimination rules. As Schroeder-Heister~\cite{Schroeder2015elim} observes:
\begin{quote}
    The intuition behind the approach based on elimination rules is that a derivation is valid, if the result of the application of each possible elimination rule to its end-formula is valid.
\end{quote}
Thus an argument is no longer valid in virtue of its form or the form to which it can be reduced (as in the introduction-based approach), but rather in virtue of the immediate consequences one can reach starting with this argument. This is a genuinely `consequentialist' view of validity.

Importantly, basing P-tV on the elimination rules does not necessarily mean that one is taking the elimination rules as prior to the introduction rules. Halln\"as and Schroeder-Heister~\cite{hallnas1990proof,hallnas1991proof,schroeder1993rules,hallnas2006proof} have shown the elimination rules arises from the introduction rules by means of \emph{Definitional Reflection} (DR):
\begin{quotation}
    whatever follows from all the defining conditions of an assertion, follows from the assertion itself
\end{quotation}
For example, disjunction ($\lor$) has the following introduction rules:
\[
\infer{\phi \lor \psi}{\phi} \qquad \infer{\phi \lor \psi}{\psi}
\]
Therefore, the defining conditions of $\phi \lor \psi$ are $\phi$ and $\psi$. Thus, DR warrants the following rule recognizable as the standard elimination rule:
\[
\infer{\chi}{\phi \lor \psi & \deduce{\chi}{[\phi]} & \deduce{\chi}{[\psi]}}
\]
Importantly, DR amounts to a \emph{closed-world assumption} --- in the sense of Reiter~\cite{reiter1981closed} --- on introduction rules as definitions.

As for P-tV based the introduction rules, it is an open problem what logic P-tV based on the elimination rules represents.  This paper shows that, assuming certain conditions about the notion of \emph{reduction} on arguments and \emph{base}, this version of P-tV corresponds to the B-eS for \emph{intuitionistic propositional logic} (IPL) by Sandqvist~\cite{Sandqvist2015IL}. 
That is, one derives the semantic clauses of the B-eS from the semantic clauses of the P-tV. Hence, this version of P-tV based on the elimination rules corresponds to IPL.

In other words, this paper says that the semantics of the logical constants (as expressed in the B-eS) is indeed given by their consequential relationships. Conversely, taking the semantics of the logical constants (as expressed in the B-eS) as conceptually prior to logical consequence, this paper shows that consequence in IPL indeed obtain by virtue of the logical form of the propositions involved.  Consequently, the elimination-based approach provides a more `consequentialist' reading of validity, focusing on the immediate uses of a proposition.

Much of the analysis in this paper concerns understanding precisely how the choice of \emph{reduction} and \emph{base} recover intuitionism. It begins with unpacking \emph{BHK},  \emph{constructivism}, and \emph{intuitionism}. This is the subject of Section~\ref{sec:main}. However, further clarity might be gained by looking at constructivism in a \emph{classical} sense. We discuss this further in Section~\ref{sec:conclusion}.

We begin in Section~\ref{sec:background} with an overview of natural deduction, the setting in which this paper takes place.  In Section~\ref{sec:ptv}, we define P-tV as used in this paper. In Section~\ref{sec:bes}, we give the B-eS for IPL by Sandqvist~\cite{Sandqvist2015IL}. The main contribution of the paper is in Section \ref{sec:main} where we formally relate P-tV and B-eS. In Section~\ref{sec:efq}, discuss the position of negation in P-tS, which is known to be a subtle issue (see, for example, K\"urbis~\cite{kurbis2019proof}, relative to the work of this paper. Finally, in Section~\ref{sec:conclusion}, we give a summary and conclusion to the paper. 

Throughout, we fix a denumerable set of atomic propositions \setAtoms. Relative to such a set we define $\setFormulas$ by the following grammar:
    \[
    \phi ::= \at{p} \in \setAtoms \mid \phi \lor \phi \mid \phi \land \phi \mid \phi \to \phi \mid \bot 
    \]
We may use meta-variables $\Gamma$ and $\Delta$ (possibly adorned with subscripts or primed) to denote sets of formulas; we use $P, Q, S$ (possibly adorned with subscripts or primed) to denote sets of atoms. We may write $\neg \phi$ to abbreviate $\phi \to \bot$.

\section{Background} \label{sec:background}

\subsection{Natural Deduction}

We require some familiarity with natural deduction in the style of Gentzen~\cite{Gentzen,Prawitz2006natural,schroeder1984natural}. In this section, we give a terse but complete summary to keep the paper self-contained. 

The objects studied in natural deduction are \emph{arguments}:

\begin{Definition}[Argument] \label{def:argument}
    An argument is a rooted, finite tree of formulas in which some (possibly none) leaves are marked as discharged. An argument is open if it has undischarged assumptions; otherwise, it is closed.
\end{Definition}

We use calligraphic style to denote arguments (e.g., $\argument{A}$ denotes an argument). The leaves of an argument $\argument{A}$ are its \emph{assumptions}, and the root is its \emph{conclusion}. An argument $\argument{A}$ is an argument from $\Gamma$ to $\phi$ iff the open assumptions of $\argument{A}$ are a subset of $\Gamma$ and the conclusion of $\argument{A}$ is $\phi$. We may use the following notations to express that $\argument{A}$ is an argument from a set of formulas $\Gamma$ to a formula $\phi$:
\[
\deduce{\phi}{\argument{A}} \qquad \deduce{\argument{A}}{\Gamma} \qquad \deduce{\phi}{\deduce{\argument{A}}{\Gamma}}
\]

Throughout this paper, we consider the composition of arguments. Let $\argument{A}$ be an argument with open assumptions $\Gamma$ and $\{\phi_1,\ldots,\phi_n\} \subseteq \Gamma$. Let $\argument{B}_1,\ldots,\argument{B}_n$ be arguments with conclusions $\phi_1,\ldots,\phi_n$, respectively. We write $\rn{cut}(\argument{B}_1,\ldots,\argument{B}_n,\argument{A})$ to denote the argument that results from composing $\argument{A}$ with $\argument{B}_1,\ldots,\argument{B}_n$ at the assumptions; that is,
\[
\rn{cut}(\argument{B}_1,\ldots,\argument{B}_n,\argument{A}) := \raisebox{-1em}{
\deduce{\argument{A}}{
    \deduce{\phi_1}{\argument{B}_1}
    \, \ldots \,
    \deduce{\phi_n}{\argument{B}_n}
}
}
\]
Note that the open assumptions of $\rn{cut}(\argument{B}_1,\ldots,\argument{B}_n)$ are $(\Gamma-\{\phi_1,\ldots,\phi_n\})\cup \Gamma_1 \cup \ldots \cup \Gamma_n$, where $\Gamma_i$ is the set of open assumptions in $\argument{B}_i$ for $i=1,\ldots,n$, respectively.

Importantly, arguments may be regulated by a set of rules. A set of rules is called a \emph{natural deduction system}. This is what we now define. We follow the presentation and ideas from Schroeder-Heister~\cite{schroeder1984natural} and Piecha and Schroeder-Heister~\cite{Schroeder2016atomic,Piecha2017definitional}.
 
\begin{Definition}[Natural Deduction Rules] \label{def:atomicrule}
 An $n$th-level rule is defined as follows:
 \begin{itemize}
  \item[-] A zeroth-level rule is a rule of the following form in which $\phi \in \setFormulas$:
 \[
 \infer{~~\phi~~}{}
 \]
 \item[-] A first-level rule is a rule of the following form in which $\phi_1,\ldots,\phi_n,\phi \in \setFormulas$,
 \[
 \infer{\,\phi\,}{\, \phi_1 & \ldots & \phi_n \,}
 \]
 \item[-] An $(n+1)$th-level rule is a rule of the following form in which $\phi_1,\ldots,\phi_n,\phi \in \setFormulas$ and $\Gamma_1,\ldots,\Gamma_n$ are (possibly empty) sets of $n$th-level atomic rules:
  \[
 \infer{\,\phi\,}{ \, \deduce{\phi_1}{[\Gamma_1]} & \ldots & \deduce{\phi_n}{[\Gamma_n]} \,}
 \]
 \end{itemize}
\end{Definition}

Having sets of rules as hypotheses is more general than having sets of propositions as hypotheses; the former captures the latter by taking zeroth-order rules. Since the premises may be empty, an $m$th-level rule is an $n$th-level rule for any $n>m$. 

%We say that a rule is \emph{properly} $n$th-level iff it is $n$th-level and at least one of the premises is a set of $(n-1)$th-level rules which are not $(n-2)$th-level rules.

\begin{Example} \label{ex:nd-rule}
The following is a natural deduction rule:
\[
\infer{~~\chi~~}{\phi_1\lor \phi_2 & \deduce{\chi}{[\phi_1]} & \deduce{\chi}{[\phi_2]}}
\] 
\end{Example}

We extract a special case of natural deduction rules that do not contain any logical constants:

\begin{Definition}[Atomic Rule] \label{ex:atomic-nd-rule}
    A rule is \emph{atomic} iff it only contains propositional variables.
\end{Definition}
\begin{Example}
    Let $p_1,p_2,d,x \in \setFormulas$. The following is an atomic rule:
    \[
    \infer{~~x~~}{r & \deduce{x}{[p_1]} & \deduce{x}{[p_2]}}
    \]
\end{Example}

Note that atomic rules are an important part of P-tS as they make up the \emph{pre-logic} notion of `proof' that forms the base case of P-tS --- see Section~\ref{sec:bes} for details. Presently, it is important to note the subtlety in their uses from other (`logical') natural deduction rules: they are not closed under substitution. This is expressed explicitly in Definition~\ref{def:derivation}. The reason is both historical and philosophical: rules containing complex formulas are logical reasoning principles appropriate for the logical constants involved; meanwhile, atomic rules are pre-logic notions of reasoning --- see Prawitz~\cite{Prawitz2006natural}, Dummett~\cite{Dummett1991logical}, and Piecha and Schroeder-Heister~\cite{Schroeder2016atomic,Piecha2017definitional}, and Schroeder-Heister~\cite{schroeder1984natural}.

\begin{Example}[Sandqvist~\cite{Sandqvist2015IL}]
    Consider the rule in Example~\ref{ex:atomic-nd-rule} with the propositions standing as follows: $r$ is `Sandy is a sibling of Mary', $p_1$ is `Sandy is a brother of Mary', $p_2$ is `Sandy is a sister of Mary', and $x$ is any other proposition. The rule expresses that the proposition $x$ may be inferred from Sandy's sibling-hood to Mary by case distinction. It is not appropriate on the basis of such reasoning being permitted that one should be able to reason by case-distinction of that sibling-hood to infer the proposition $x$ from some statement, say, `The sky is blue.'
\end{Example}

A collection of rules is called a \emph{system}:

\begin{Definition}[Natural Deduction System]
    A natural deduction system is a set of natural deduction rules. 
\end{Definition}
\begin{Definition}[Atomic System]
    An \emph{atomic system} is a natural deduction system comprising only atomic rules.
\end{Definition}

We use script-style to denote atomic systems (e.g., $\base{A}$ denotes an atomic system). 

\begin{Example}
    In Figure~\ref{fig:nj} is shown the natural deduction system \calculus{NJ} by Gentzen~\cite{Gentzen}. 
    %Observe that it is properly second-level and that all the hypotheses in it are zeroth-level rules.
\end{Example}

While a natural deduction system may have infinitely many rules, it is at most countably infinite as the language is countable. %A natural deduction system $\system{N}$ is \emph{properly} $n$th-level iff, for any $\rn{r} \in \system{N}$, there is $k \leq n$ such that $\rn{r}$ is properly $k$th-level.

An argument regulated by a natural deduction system $\system{N}$ is called an \emph{$\system{N}$-derivation}. The set of  \emph{$\system{N}$-derivations} may be defined inductively by composing instances of rules from $\base{N}$. To define this formally, we require substitution:

\begin{Definition}[Substitution Function]
    A substitution function is a mapping $\theta:\setAtoms \to \setFormulas$. The set of all substitutions is $\setSubs$.
\end{Definition}

Typically, one defines a substitution function by specifying the map for some finite subsection of $\setAtoms$ and extending it to the rest of the domain by some arbitrary assignment to formulas. This is a minor detail that does not affect the work in this paper. The action of a substitution $\theta$ extends to formulas as follows:
\[
\theta(\phi) :=
 \begin{cases}
\theta(\at{p}) & \text{if } \phi = \at{p} \in \setAtoms \\
\bot & \text{if } \phi =  \bot \\
\theta(\psi_1) \circ \theta(\psi_2) & \text{if } \phi = \psi_1 \circ \psi_2 \text{ for } \circ \in \{\to, \land, \lor\} \\
\end{cases}
\]

\begin{Definition}[Derivation in a Natural Deduction System] \label{def:derivation}
 Let $\system{N}$ be a natural deduction system. The set of $\system{N}$-derivations is defined inductively as follows:
 \begin{itemize}[label=--]
     \item \textsc{Base Case}. Let $\rn{r} \in \system{N}$ be a zeroth-level rule concluding $\phi$. We consider two cases:
     \begin{itemize}[label=--]
     \item $\rn{r}$ is atomic.
     The natural deduction argument consisting of the node $\phi$ is a $\system{N}$-derivation.
     \item $\rn{r}$ is \emph{not} atomic. For any substitution $\theta$, 
      the node $\theta(\phi)$ is an $\system{N}$-derivation.
      \end{itemize}
     \item \textsc{Induction Step}. Let $\rn{r} \in \system{N}$ be an $(n+1)$st-level rule,
  \[
 \infer{\,\phi\,}{ \, \deduce{\phi_1}{[\Gamma_1]} & \hdots & \deduce{\phi_n}{[\Gamma_n]} \,}
 \]
 We consider two cases:
 \begin{itemize}
     \item $\rn{r}$ is atomic. Suppose for each $1 \leq i \leq n$ there is an $\system{N}$-derivation $\argument{D}_i$ of the following form:
 \[
 \deduce{\phi_i}{\deduce{\mathcal{D}_i}{\Gamma_i, \Delta_i}}
 \]
     The natural deduction argument with root $\phi$ and immediate sub-trees $\argument{D}_1, \ldots, \argument{D}_n$ is a $\system{N}$-derivation of $\phi$ from $\Delta_1 \cup \ldots \cup \Delta_n$.
\item $\rn{r}$ is \emph{not} atomic. Suppose for any substitution $\sigma$ such that for each $1 \leq i \leq n$ there is an $\system{N}$-derivation $\argument{D}_i$ of the following form:
 \[
 \deduce{\sigma(\phi_i)}{\deduce{\mathcal{D}_i}{\sigma(\Gamma_i), \Delta_i}}
 \]
 The natural deduction argument with root $\sigma(\phi)$ and immediate sub-trees $\argument{D}_1, \ldots, \argument{D}_n$ is a $\system{N}$-derivation of $\phi$ from $\Delta_1 \cup \ldots \cup \Delta_n$.
 \end{itemize}
 \end{itemize}
 \end{Definition}

\begin{Definition}[Derivability] \label{def:derivability}
     Let $\system{N}$ be a natural deduction system. The $\system{N}$-derivability relation $\proves_\system{N}$ is defined as follows: 
     \[
     \begin{array}{lcl}
     \Gamma \proves_\system{N} \phi & \text{iff} & \text{there exists an $\system{N}$-derivation $\argument{D}$ such that} \\ & & \mbox{the open assumptions of $\argument{D}$ are subset of $\Gamma$ and the conclusion is $\phi$}
     \end{array}
     \]
\end{Definition}

% The difference in treatment between atomic systems and non-atomic systems means that one may desire to have parallel definitions of derivability rather than a combined one. However, at certain points in Section~\ref{sec:bes}, Section~\ref{sec:ptv}, and Section~\ref{sec:main}, we will work in systems that have both atomic and non-atomic rules, thus requiring a combined treatment.

An $\calculus{N}$-derivation is \emph{closed} iff it is closed as an argument, in which case it is called an $\calculus{N}$-proof.

\begin{Example}
    The following is an example of an open $\calculus{NJ}$-derivation:
    \[
    \infer[\ern \lor]{q \lor p}{p \lor q & \infer[\irn \lor]{q \lor p}{[p]} & \infer[\irn \lor]{q \lor p}{[q]}}
    \]
    It witnesses $p \lor q \proves_{\calculus{NJ}} q \lor p$. The labels on the inferences are to aid readability and are not formally part of the argument. 
\end{Example}

This concludes a general introduction to natural deduction in the sense of Gentzen~\cite{Gentzen} suitable for this paper. We include some further \emph{specific} background of natural deduction --- namely, normalization results for $\system{NJ}$ --- below.

\subsection{System $\calculus{NJ}$} \label{sec:il}

There are many presentations of IPL in the literature. Therefore, we begin by fixing the relevant concepts and terminology for this paper.

In this paper, IPL is a certain consequence judgement $\proves$ on sequents. 
%We write $\Gamma \proves \phi$ to denote that $\Gamma \seq \phi$ is a consequence of IPL. 
Our principal characterization will be through a natural deduction system.

\begin{Definition}[Natural Deduction System \calculus{NJ}]
    The natural deduction system $\calculus{NJ}$ is comprised of the rules in Figure~\ref{fig:nj}.
\end{Definition}

\begin{figure}[t]
\figureline
\vspace{1mm}
\[
\begin{array}{c}
       \infer[\irn{\land}]{\phi \land \psi}{\phi & \psi} \qquad \infer[\ern{\land^1}]{\phi}{\phi \land \psi} \quad \infer[\ern{\land^2}]{\psi}{\phi \land \psi}
       \qquad \infer[\irn{\to}]{\phi \to \psi}{\deduce{\phi}{[\psi]}} 
       \qquad
       \infer[\ern{\bot}]{\phi}{\bot} 
\\[4mm]
    \infer[\irn{\lor^1}]{\phi \lor \psi}{\phi} 
    \quad 
    \infer[\irn{\lor^2}]{\phi \lor \psi}{\psi} 
    \qquad
    \infer[\ern{\lor}]{\chi}{\phi \lor \psi & \deduce{\chi}{[\phi]} & \deduce{\chi}{[\psi]}}
    \qquad
    \infer[\ern{\to}]{\phi}{\phi & \phi \to \psi} 
\end{array}
\]
\figureline
\caption{Natural Deduction System $\calculus{NJ}$}
    \label{fig:nj}
\end{figure}

\begin{Proposition}[Gentzen~\cite{Gentzen}] \label{prop:njproof}
     There is an \calculus{NJ}-proof of $\phi$ iff $\emptyset \proves \phi$.
\end{Proposition}

% A restatement (using the Deduction Theorem --- see Herbrand~\cite{herbrand1930recherches}) more useful for the work in this paper is as follows:

% \begin{Proposition}\label{prop:nj}
%      There is an \calculus{NJ}-derivation from $\Gamma$ to $\phi$ iff $\Gamma \proves \phi$.
% \end{Proposition}

The rules of $\calculus{NJ}$ with subscripts $\rn{I}$ and $\rn{E}$ are the \emph{introduction rules} ($I$-rules) and \emph{elimination rules} ($E$-rules), respectively. They sometimes come in pairs; for example,
\[
\infer[\ern{\land}]{
        \phi
    }{
    \infer[\irn{\land}]{\phi \land \psi}{\deduce{\phi}{\mathcal{D}_1} & \deduce{\psi}{\mathcal{D}_2}}
    }
\]
Such derivations contain superfluous argumentation for $\phi$ and so are called \emph{detours}.

\begin{Definition}[Detour]
    A \emph{detour} in a derivation is a sub-derivation in which a formula is obtained by an $I$-rule and is then the major premise of the corresponding $E$-rule.
\end{Definition} 

\begin{Definition}[Canonical]
    A derivation is canonical iff it contains no detours.
\end{Definition}

Prawitz~\cite{Prawitz2006natural} proved that canonical \calculus{NJ}-proofs are complete for IPL; that is, we may refine Proposition~\ref{prop:njproof} as follows:

\begin{Proposition}[Prawitz~\cite{Prawitz2006natural}] \label{prop:normal}
     There is a canonical \calculus{NJ}-derivation from $\Gamma$ to $\phi$ iff $\Gamma \proves \phi$.
\end{Proposition}

The argument uses a reduction relation $\redto$ that precisely eliminates detours; for example, detours with implication ($\to$) are reduced as follows:
\[
\infer[\ern{\to}]{\psi}{
    \deduce{\phi}{\mathcal{D}_1} & 
    \infer[\irn{\to}]{\phi \to \psi}
        {
            \deduce{\psi}{\deduce{\mathcal{D}_2}{[\phi]}}
        }
    }
    \qquad
    \raisebox{1em}{$\redto$}
    \qquad
    \deduce{\psi}{
         \deduce{\mathcal{D}_2}{
            \deduce{\phi}{\mathcal{D}_1}
            }
        }
\]
The reflexive and transitive closure of $\redto$ is denoted $\redtostar$. This reduction relation is normalizing, and its normal forms are canonical proofs. 

% \begin{Proposition}[Prawitz~\cite{Prawitz2006natural}] \label{prop:normalization}
%     If $\argument{A}$ is an $\calculus{NJ}$-derivation from $\Gamma$ to $\phi$, then there is a canonical $\calculus{NJ}$-derivation $\argument{A}'$ from $\Gamma$ to $\phi$ such that $\argument{A} \redtostar \argument{A}'$.
% \end{Proposition}
% \begin{Corollary}[Prawitz~\cite{Prawitz2006natural}]  \label{cor:intro}
%     There is a canonical \calculus{NJ}-derivation $\argument{A}$ from $\Gamma$ to $\phi$ that concludes by the use of an introduction rule iff $\Gamma \proves \phi$.
% \end{Corollary}

This establishes the relevant syntax and proof theory required for IPL in this paper.

\subsection{The BHK Interpretation} \label{sec:bhk}
Intuitionism, as defined by Brouwer~\cite{brouwer1913intuitionism}, is the view that an argument is valid when it provides sufficient evidence for its conclusion. This is IL. Famously, as a consequence, IL rejects \emph{the law of the excluded middle} --- that is, the meta-theoretic statement that either a statement or its negation is valid. This law is equivalent to the principle that in order to prove a proposition it suffices to show that its negation is contradictory. In IL, such an argument does not constitute sufficient evidence for its conclusion. 

Heyting~\cite{heyting1966intuitionism} and Kolmogorov~\cite{kolmogorov} provided a semantics for intuitionistic proof that captures the evidential character of intuitionism, called the Brouwer-Heyting-Kolmogorov (BHK) interpretation of IL. It is now the standard explanation of the logic.

Supposing a notion of proof for atomic formulae,
\begin{itemize}[label=--]
    \item a proof $\argument{A}$ of $\phi \land \psi$ is a pair $\langle \argument{B}_1, \argument{B}_2 \rangle$ such that $\argument{B}_1$ is a proof of $\phi$ and $\argument{B}_2$ is a proof of $\psi$;
    \item a proof $\argument{A}$ of $\phi \lor \psi$ is either a pair $\langle 0, \argument{B} \rangle$ such that $\argument{B}$ is a proof of $\phi$ or a pair $\langle 1, \argument{B} \rangle$ such that $\argument{B}$ is a proof of $\psi$;
    \item a proof of $\phi \to \psi$ is a method $f$ for constructing a proof of $\psi$ from a proof of $\phi$;
    \item nothing is a proof of $\bot$.
\end{itemize}

The \emph{propositions-as-types} correspondence --- see Howard, Barendregt, and others \cite{howard1980formulae,BDS2013,Barendregt1991,Barendregt1993} --- gives a standard way of instantiating the BHK interpretation as terms in the simply-typed $\lambda$-calculus. Technically, the set-up can be sketched as follows: a judgement 
% \[
%      \phi_1 , \ldots , \phi_k  \vdash \Phi \,:\, \phi  
% \]
that $\mathcal{D}$ is an $\calculus{NJ}$-proof of the sequent $\phi_1 , \ldots , \phi_k \seq \phi$ corresponds to a typing 
judgement 
\[
    x_1 : A_1 , \ldots , x_k : A_k \vdash M(x_1, \ldots, x_k) : A 
\]
where the $A_i$s are types corresponding to the $\phi_i$s, the $x_i$s correspond to placeholders for proofs of the $\phi_i$s, the $\lambda$-term $M(x_1, \ldots, x_k)$ corresponds to $\mathcal{D}$, and the type $A$ corresponds to $\phi$. 

Lambek~\cite{lambek1980lambda} gave a more abstract account by showing that simply-typed $\lambda$-calculus is the internal language of \emph{cartesian closed categories} (CCCs), thereby giving a categorical semantics of proofs for IPL. In this set-up, 
a morphism 
\[
    \llbracket \phi_1 \rrbracket \times \ldots \times \llbracket \phi_k \rrbracket 
        \stackrel{\llbracket \mathcal{D} \rrbracket}{\longrightarrow} 
            \llbracket \phi \rrbracket
\]
in a CCC, where $\times$ denotes cartesian product, that interprets the $\calculus{NJ}$-proof $\mathcal{D}$ of $\phi_1 , \ldots , \phi_k \seq \phi$ also interprets the term $M$, where the $\llbracket \phi_i \rrbracket$s interpret also the $A_i$s and $\llbracket \phi \rrbracket$ also interprets $A$. 

Altogether, this describes the \emph{Curry-Howard-Lambek} correspondence for IPL as summarized in Figure~\ref{fig:chl} in which:
\begin{itemize}[label=--]
    \item $\mathcal{D} \Rightarrow \Gamma \seq \phi$ denotes that $\mathcal{D}$ is an $\calculus{NJ}$-derivation of $\phi$ from $\Gamma$;
    \item $x : A_\Gamma \vdash M(x) : A_\phi$ denotes a typing judgment, as described above, corresponding to $\mathcal{D}$; and
    \item $\llbracket \Gamma \rrbracket \stackrel{\llbracket \mathcal{D} \rrbracket}{\rightsquigarrow} \llbracket \phi \rrbracket$ denotes that $\llbracket \mathcal{D} \rrbracket$ is a morphism from $\llbracket \Gamma \rrbracket$ to $\llbracket \phi \rrbracket$ in a CCC. 
\end{itemize} 

\begin{figure}[t]
    \hrule
    \[
\xymatrix{
    x : \Gamma \proves M(x) : \phi \ar@{<->}[rr] & & \llbracket \Gamma \rrbracket \stackrel{\llbracket \mathcal{D} \rrbracket}{\rightsquigarrow} \llbracket \phi \rrbracket \\
    & \ar@{<->}[ul] \mathcal{D} \Rightarrow \Gamma \seq \phi \ar@{<->}[ur] &
}
\]
\hrule \vspace{1mm}
    \caption{Curry-Howard-Lambek Correspondence}
    \label{fig:chl}
\end{figure}

To generalize to full IL (and beyond), Seely~\cite{seely1983hyperdoctrines} modified this categorical set-up and introduced \emph{hyperdoctrines} --- indexed categories of CCCs with coproducts over a base with finite products. Martin-L\"of~\cite{martin1975intuitionistic} gave a formulae-as-types correspondence for predicate logic using dependent type theory. Barendregt~\cite{Barendregt1991} gave a systematic treatment of type systems and the propositions-as-types correspondence. A categorical treatment of dependent types came with Cartmell~\cite{cartmell} --- see also, for examples among many, work by Streicher~\cite{Streicher1988}, Pavlovi\'c~\cite{Pavlovic1990}, Jacobs~\cite{Jacobs}, and Hofmann~\cite{Hofmann1997}. In total, this gives a semantic account of \emph{proof} for first- and higher-order predicate intuitionistic logic based on the BHK interpretation.

Pym and Ritter~\cite{pym2004reductive} have provided a generalization of the BHK interpretation closely related to P-tV. The original motivation was as a semantics of reductive logic.

The traditional approach to logic is through the \emph{deductive} paradigm in which conclusions are inferred from established premises. However, in practice, logic typically proceeds through the dual paradigm known as \emph{reductive} logic: sufficient premises are inferred from putative conclusions by means of `backwards' inference rules. Pym and Ritter~\cite{pym2004reductive} have given a semantics to constructions in reductive logic using the language of BHK in a way that recalls P-tV. Specifically, they give the \emph{constructions-as-realizers-as-arrows correspondence}~\ref{fig:cor} based on \emph{polynomial} categories, which extend a category in which arrows denote proofs for a logic by additional arrows that supply `proofs' for propositions that do not have proofs but appear during reduction:
\begin{enumerate}[label=--]
    \item $\Phi \Rightarrow \Gamma \seq \phi$ denotes that $\Phi$ is a sequence of reductions for the sequent $\Gamma \seq \phi$;
    \item $[\Gamma] \proves [\Phi] : [\phi]$ denotes that $[\Phi]$ is a \emph{realizer} of $[\phi]$ with respect to the assumptions $[\Gamma]$; and
    \item $\llbracket \Gamma \rrbracket \stackrel{\llbracket \Phi \rrbracket}{\rightsquigarrow} \llbracket \phi \rrbracket$ denotes that $\llbracket \Gamma \rrbracket$ is a morphism from $\llbracket \Gamma \rrbracket$ to $\llbracket \phi \rrbracket$ in the appropriate polynomial category. 
\end{enumerate}

They also defined a judgement 
$w \Vdash_\Theta (\Phi:\phi)\Gamma$  which says that $w$ is a world witnessing that $\Phi$ is a reduction of $\phi$ to $\Gamma$, relative to the indeterminates of $\Theta$.

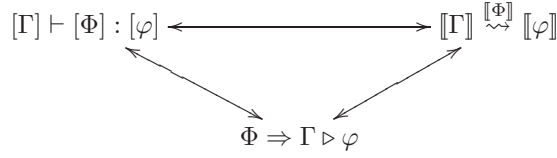
\begin{figure}[t]
    \hrule
    \[
\xymatrix{
    [\Gamma] \proves [\Phi]:[\phi] \ar@{<->}[rr] & &   \llbracket \Gamma \rrbracket \stackrel{\llbracket \Phi \rrbracket}{\rightsquigarrow} \llbracket \phi \rrbracket \\
    & \ar@{<->}[ul] \Phi \Rightarrow \Gamma \seq \phi \ar@{<->}[ur] &
}
\]
\hrule \vspace{1mm}
\caption{Constructions-as-Realizers-as-Arrows Correspondence}
    \label{fig:cor}
\end{figure}

We observe in this characterization obvious similarities with the BHK interpretation of IL (see Section~\ref{sec:il}). Specifically, it coheres with the generalization by Pym and Ritter~\cite{pym2004reductive}:
\begin{itemize}[label=--]
    \item the judgment $\Phi \Rightarrow \Gamma \seq \phi$ corresponds to P-tV, 
    \item the judgment $[\Gamma] \proves [\Phi]:[\phi]$ corresponds to the realizers interpretation of arguments in Section~\ref{sec:main}, and
    \item the judgment $\llbracket \Gamma \rrbracket \stackrel{\llbracket \Phi \rrbracket}{\rightsquigarrow} \llbracket \phi \rrbracket$ corresponds to B-eS.
    \end{itemize}
Thus, in this paper, we move from the realizers perspective, in which the witnessing arguments must be constructed explicitly to the types perspective in which the witnessing arguments are observed implicitly as arrows. %More precisely, the arrows characterize an inductively defined judgment \mbox{$W \entails_\Theta (\Phi:\phi) \Gamma$} in which $W$ is a state of knowledge (i.e., the analogue of a base) and $\Theta$ is a set of indeterminates.

\section{Proof-theoretic Validity} \label{sec:ptv}

There are several accounts of proof-theoretic validity—see, for example, Prawitz~\cite{Prawitz2006natural,Prawitz1971ideas,Prawitz1973towards,Prawitz1972} and Schroeder-Heister~\cite{Schroeder2006validity,SEP-PtS}. Typically, they are based on the introduction rule and can be seen as a way to realize the sentiment expressed by Gentzen~\cite{Gentzen} (see Section~\ref{sec:intro}) that the introduction rules represent definitions of the connectives involved. Prawitz's Conjecture~\cite{Prawitz1973towards} is the statement that IPL is complete with respect to P-tV based on the introduction rules.

As explained in Section~\ref{sec:intro}, Piecha et al.~\cite{Piecha2015failure} show that Prawitz's Conjecture fails when P-tV is slightly simplified. Indeed, Stafford~\cite{Stafford2021} has shown that the logic described by the simplified version of P-tV is a \emph{(general) inquisitive logic}. In contrast, Takemura~\cite{Takemura} suggests that the conjecture holds without the simplification.

In this paper, we consider an alternate version of P-tV that is based on the elimination rules. In contrast to the treatment of P-tV based on the introduction rules, we show that this version of P-tV does correspond to IPL when certain details of the setup are met.

In Section~\ref{sec:background}, we introduced the idea of an atomic system. These atomic systems will form the base case of validity. However, we do not necessarily want to consider \emph{all} atomic systems, but rather some specific forms—see, for example, Piecha and Schroeder-Heister~\cite{Schroeder2016atomic,Piecha2017definitional}. Therefore, fix a notion of base $\mathfrak{B}$ (i.e., a subset of all atomic systems). It is these systems that will form the basis of validity. Henceforth, we consider only atomic systems that are bases. Thus, given a base $\base{B}$, we may write $\base{C} \supseteq \base{B}$ to express that $\base{C}$ is a superset of $\base{B}$ that is also a base.

Furthermore, in Section~\ref{sec:background}, we discussed the idea of reduction by Prawitz~\cite{Prawitz2006natural}. This also forms an essential part of the definition of proof-theoretic validity. Presently, we will not fix some particular set of reductions, but rather work relative to such sets as a parameter. Thus, let $\mathbb{R}$ be a set of reductions—that is, functions from arguments to arguments. While specific features of the elements of $\mathbb{R}$ may be desirable (e.g., that they are computable), we shall not impose any such restrictions until they become necessary.

\begin{Definition}[Validity for Arguments] \label{def:bvalid}
   Let $\base{B} \in \mathfrak{B}$ be a base. An argument $\argument{A}$ is $\base{B}$-valid iff:
    \begin{enumerate}[label={(\arabic*)}]
      \item it is a closed argument ending with an atomic formula and either it is a $\base{B}$-proof or it $\mathbb{R}$-\emph{reduces} to a $\calculus{NJ}\cup\base{B}$-proof\label{def:bvalid:one}
        \item it is a closed argument ending with a complex formula and, for any $\base{C} \supseteq \base{B}$, for any extension of $\mathcal{A}$ by an elimination rule, using $\base{C}$-valid arguments for the minor derivations and restricting to atomic conclusions where applicable (namely, $\ern \bot$ and $\ern{\lor}$), the result is a $\base{C}$-valid argument \label{def:bvalid:two}
        \item it is an open argument such that, for any $\base{C} \supseteq \base{B}$, any extension of $\argument{A}$ by $\base{C}$-valid arguments of the assumptions $\argument{C}_1, \ldots, \argument{C}_n$ results in a $\base{C}$-valid argument. \label{def:bvalid:three}
    \end{enumerate}
    An argument $\argument{A}$ is valid iff $\argument{A}$ is $\base{B}$-valid for every base $\base{B}$.
\end{Definition}
We may write $\Gamma \entails_{\base{B}} \phi$ to denote that there exists a $\base{B}$-valid argument from $\Gamma$ to $\phi$. Similarly, we may write $\Gamma \entails \phi$ to denote that there is a valid argument from $\Gamma$ to $\phi$. Immediately by Definition~\ref{def:bvalid}, 
\[
\Gamma \entails \phi \qquad \text{iff} \qquad \Gamma \entails_{\base{B}} \phi
\]

Definition~\ref{def:bvalid} merits some remarks, particularly \ref{def:bvalid:two}. The restriction to $\base{C}$-derivations with an \emph{atomic} conclusion is required to render the definition inductive. In practice, this side condition only arises in the case of $\ern \lor$ and $\ern \bot$. However, one could replace $\ern \land^1$ and $\ern \land^2$ with
\[
\infer{\chi}{\phi \land \psi & \deduce{\chi}{[\phi]} & \deduce{\chi}{[\psi]}}
\]
and the condition would then be more uniformly applied. The original formulation of validity for arguments based on elimination-rules by Prawitz~\cite{Prawitz1971ideas} did not have the restriction and instead omitted disjunction altogether. The inclusion of disjunction was published in Prawitz~\cite{Prawitz2007}, where he also refers to Dummett~\cite{Dummett1991logical}. 

Schroeder-Heister~\cite{Schroeder2006validity} observe that the restriction is closely related to the fact that the definability of first-order logical constants in second-order propositional $\forall \to$-logic. In particular,
\[
U \lor V := \forall X. (U \to X) \to (U \to X) \to X
\]
--- see Prawitz~\cite{Prawitz2006natural}. This suggests that proof-theoretic validity based on the elimination rules corresponds to \emph{atomic second-order propositional logic} $F_\text{at}$, as studied by Ferreira~\cite{Ferreira2006}. In this case, the variable $X$ in the second-order definition of $\lor$ is restricted to range over atoms.

There is a question regarding what logic this notion of P-tV represents. Ferreira and Ferreira~\cite{Ferreria2013,ferreira2014faithfulness} showed that with the second-order definitions, $F_\text{at}$ corresponds to IPL. Schroeder-Heister~\cite{Schroeder2015elim} notes that this result renders IPL a good candidate for the notion of validity above, but that there is room for doubt as it depends on the choice of $\mathfrak{B}$ and $\mathbb{R}$. In Section~\ref{sec:main}, we specify some condition for which the notion of validity in Definition~\ref{def:bvalid} characterizes IPL.

% That the restriction to atoms in the treatment of $\bot$ and $\lor$ enables generalization to all formulae is called \emph{instantiation overflow}.

Taking $\mathbb{R}$ to only contain Prawitz's~\cite{Prawitz2006natural} reduction operators does not suffice to render all the rules of \calculus{NJ} valid.

\begin{Example}[Schroeder-Heister~\cite{Schroeder2015elim}] \label{ex:permutations}
    Consider $\ern \lor$ for any $\base{B}$,
    \[
    \text{if $\Gamma \entails_{\base{B}} \phi \lor \psi$, $\Gamma, \phi \entails_{\base{B}} \chi$, and  $\Gamma, \psi \entails_{\base{B}} \chi$, then  $\Gamma \entails_{\base{B}} \chi$} \tag{$\entails_\base{B} \ern \lor$}
    \]
    To render this statement true, it suffices to use, in addition to Prawitz's~\cite{Prawitz2006natural} reductions, some \emph{permutative} reductions.
    
For example, in the case where $\chi = \chi_1 \land \chi_2$, one has the argument
\[
\infer{\chi_1 \land \chi_2}{\deduce{\phi \lor \psi}{\mathcal{D}} & \deduce{\chi_1\land\chi_2}{\deduce{\mathcal{D}_1}{[\phi]}} & \deduce{\chi_1\land\chi_2}{\deduce{\mathcal{D}_2}{[\psi]}} }
\]
where $\argument{D}$, $\argument{D}_1$, $\argument{D}_2$ witness the hypotheses in $(\entails_\base{B} \ern \lor)$.  

To show that this argument is $\base{B}$-valid, given that $\mathcal{D}$, $\mathcal{D}_1$ and $\mathcal{D}_2$ are $\base{B}$-valid, we may use the following \emph{permutation} reduction:
\[
\infer{\chi_i}{\infer{\chi_1 \land \chi_2}{\deduce{\phi \lor \psi}{\mathcal{D}} & \deduce{\chi_1\land\chi_2}{\deduce{\mathcal{D}_1}{[\phi]}} & \deduce{\chi_1\land\chi_2}{\deduce{\mathcal{D}_2}{[\phi]}}}
}
\qquad \raisebox{2em}{$\redto_i$} \qquad
\infer{\chi_i}{\deduce{\phi \lor \psi}{\mathcal{D}} & \infer{\chi_i}{\deduce{\chi_1\land\chi_2}{\deduce{\mathcal{D}_1}{[\phi]}}} & \infer{\chi_i}{\deduce{\chi_1\land\chi_2}{\deduce{\mathcal{D}_2}{[\psi]}}}}
\]
We carry on reducing until the conclusion is an atom, at which point we appeal to \ref{def:bvalid:one} in Definition~\ref{def:bvalid}. 
\end{Example}

We remain agnostic as to exactly what reductions are used:

\begin{Condition} \label{cond:disj-elim}
    The set of reductions $\mathbb{R}$ is such that the use of disjunction elimination (i.e., $\ern \lor$) is a $\base{B}$-valid argument.

    That is, the argument 
     \[
    \infer{\chi}{
   \chi_1 \lor \chi_2 & \ddeduce{\chi}{\argument{B}_1}{\Gamma, [\chi_1]} & \ddeduce{\chi}{\argument{B}_2}{\Gamma, [\chi_2]}}
    \]
    is $\base{B}$-valid whenever $\argument{B}_1$ and $\argument{B}_2$ are $\base{B}$-valid.
    
   Furthermore, for any $\base{C} \supseteq \base{B}$, 
    \[
    \infer{\chi}{
    \ddeduce{\chi_1 \lor \chi_2}{\argument{B}}{\Gamma} & \ddeduce{\chi}{\argument{B}_1}{\Gamma, [\chi_1]} & \ddeduce{\chi}{\argument{B}_2}{\Gamma, [\chi_2]}}
    \]
    is $\base{C}$-valid whenever $\argument{B}$ is $\base{B}$-valid, and $\argument{B}_1$ and $\argument{B}_2$ are $\base{C}$-valid.
\end{Condition}

We return to the issue of reduction operators and bases in Section~\ref{sec:main}. Presently, we require only that the notion of reduction can eliminate any extraneous logical content from a $\base{B}$-valid argument:

\begin{Condition}\label{cond:atomic}
    If there is a $\base{B}$-valid argument from $P \subseteq \mathbb{A}$ to $p \in \mathbb{A}$, then there is a $\base{B}$-derivation from $P$ to $p$.
\end{Condition}

The same idea allows us to characterize exactly what it means for a formula $\chi$ to be a vacuous assumption in an argument. This means either that there already exists a valid argument for it relative to the current assumptions, so that it needs not be assumed itself; or it means that whenever there is a valid argument from it, there is a valid argument without it. These two cases correspond to the left and right directions of the following condition:

\begin{Condition} \label{cond:comp}
    The set of reductions $\mathbb{R}$ is such that the following holds:
    \[
    \Gamma \entails_{\base{B}} \chi \qquad \mbox{iff} \qquad \mbox{$\forall \base{C} \supseteq \base{B}, \forall p \in \setAtoms$, if $\chi, \Gamma \entails_{\base{C}} p$, then $\Gamma \entails_{\base{C}} p$}
    \]
\end{Condition}

This completes the definition of P-tV.

\section{Base-extension Semantics} \label{sec:bes}

As explained in Section~\ref{sec:intro}, by `base-extension semantics' we mean semantics for a logic in terms of proofs in a base. There are many varied examples of B-eS in the literature including, \emph{inter alia}, Eckhardt and Pym~\cite{Eckhardt} for modal logic, Gheorghiu et al.~\cite{IMLL,BI} for substructural logic, Makinson~\cite{makinson2014inferential} for classical logic, Nascimento~\cite{nascimentothesis,stafford2023} using multi-bases, and Piecha et al.~\cite{Piecha2015failure,Piecha2016completeness,Piecha2019incompleteness} and Stafford~\cite{Stafford2021} for super-intuitionistic logics. Also related, taking a slightly more general scope, is work by Goldfarb~\cite{goldfarb2016dummett}. In this paper, we follow the version of B-eS by Sandqvist~\cite{Sandqvist2009CL,Sandqvist2015IL}.

The B-eS for IPL given by Sandqvist~\cite{Sandqvist2015IL} only admits certain atomic systems. A \emph{Sandqvist Base} is a an atomic system containing only atomic rules of the form
\[
\infer{~~\at{c}~~}{} \qquad  \infer{c}{p_1 & \ldots & p_n} \qquad \infer{\at{c}}{~~ \deduce{\at{p}_{1}}{[\at{Q}_{1}]} & \ldots & \deduce{\at{p}_n}{[\at{Q}_n]}
~~}
\]
where $Q_1,\ldots,Q_n$ are (possibly empty) finite sets of atoms. The set of Sandqvist bases is denoted $\mathfrak{S}$.

% Observe that relative to Sandqvist Bases, derivability in $\calculus{NJ} \cup \base{B}$ is understood simply as a natural deduction in the sense of Gentzen~\cite{Gentzen} in the combined system --- see Definition~\ref{def:derivation}.

The B-eS for IPL is defined by the following \emph{support} judgment:

\begin{Definition}[Support in a Base] \label{def:supp}
    Support in a base $\base{B}$ is the least relation $\supp_\base{B}$ satisfying the clauses of Figure~\ref{fig:sandqvist:support}, where $\Delta \neq \emptyset$.
\end{Definition}

\begin{figure}[t]
\hrule 
\vspace{2mm}
\[
    \begin{array}{l@{\quad}c@{\quad}l@{\quad}r}
        \supp_{\base{B}} \at{p}  & \mbox{iff} &   \proves_{\base{B}} \at{p} & \mbox{(At)}  \\[1mm]
        \supp_{\base{B}} \bot & \mbox{iff} &    \supp_{\base{B}} \at{p} \text{ for any } \at{p} \in \setAtoms & (\bot) \\[1mm]
        \supp_{\base{B}} \phi \land \psi   & \mbox{iff} & \mbox{$\supp_{\base{B}} \phi$ and $\supp_{\base{B}} \psi$} & (\land) \\[1mm]
        \supp_{\base{B}} \phi \lor \psi & \mbox{iff} &  \mbox{$\forall \base{C} \supseteq \baseB$ and $\forall \at{p} \in \setAtoms$,} \\ 
        & & \mbox{if $\phi \supp_{\base{C}} \at{p}$ and $\psi \supp_{\base{C}} \at{p}$, then $\supp_{\base{C}} \at{p}$} & (\lor)  \\[1mm]
        \supp_{\base{B}} \phi \to \psi & \mbox{iff} & \phi \supp_{\base{B}} \psi & (\to) \\[1mm]
        \hspace{-4mm} \Delta \supp_{\base{B}} \phi & \mbox{iff} & \mbox{$\forall \base{C} \supseteq \baseB$, if $\supp_{\base{C}} \psi$ for any $\psi \in \Delta$, then $\supp_{\base{C}} \phi$ } &  \mbox{(Inf)} \\[1mm]
        \hspace{-1em} \Gamma \supp_{} \phi & \mbox{iff} & \mbox{$\Gamma \supp_{\base{B}} \phi$ for any $\base{B}$}
    \end{array}
\]
\hrule
\caption{Base-extension Semantics for IPL} 
\label{fig:sandqvist:support}
\end{figure}

There are clear similarities between the B-eS  and P-tV. Perhaps most striking is the treatment of disjunction ($\lor$). In the case of the B-eS, the `second-order' definition adumbrates the categorical expression of P-tS by Pym et al.~\cite{Pym2022catpts,pym2024categorical} and the logic programming perspective by Gheorghiu and Pym~\cite{NAF}. Importantly, Piecha et al.~\cite{Piecha2015failure,Piecha2016completeness,Piecha2019incompleteness} have shown that the semantics obtained by replacing ($\lor$) in Figure~\ref{fig:sandqvist:support} with a Kripke-like clause for disjunction 
\[
\supp_{\base{B}} \phi \lor \psi \qquad \mbox{iff} \qquad \supp_{\base{B}} \phi \text{ or } \supp_{\base{B}} \psi
\]
is incomplete for IPL.

\begin{Theorem}[Sandqvist~\cite{Sandqvist2015IL}]\label{thm:sandqvist:snc}
    $\Gamma \supp \phi$ iff $\Gamma \proves_{\base{B}} \phi$.
\end{Theorem}

We want to relate this semantics directly to the construction of intuitionistically valid arguments. The theorem cannot be generalized as follows:
\[
\Gamma \supp_\base{B} \phi \qquad \mbox{iff} \qquad \Gamma \proves_{\calculus{NJ} \cup \base{B}} \phi
\]
An immediate counter-example is given in the case where $\Gamma = \emptyset$ and $\phi = \bot$. \smallskip

\noindent Aside: \emph{In this paper, $\bot$ cannot be proved in a base $\base{B}$ --- that is, $\proves_{\base{B}} \bot$ is impossible. Indeed, it is not grammatical as bases only concern non-logical content and, in this paper, $\bot$ is regarded as a logical construct! While there are versions of P-tS in which 
$\bot$ is included in bases --- see, for example, Piecha et al.~\cite{Piecha2015failure} --- we follow the set-up of Prawitz~\cite{Prawitz1971ideas} and Sandqvist~\cite{sandqvist2015hypothesis} where it is not.} \smallskip

Nonetheless, the base case holds and plays an important part in the proof of completeness:
\begin{Lemma}[Sandqvist~\cite{Sandqvist2015IL}] \label{lem:sand-base}
    Let $P \subseteq \setAtoms$ and $p \in \setAtoms$,
    \[
    P \supp_{\base{B}} p \qquad \mbox{iff} \qquad P \proves_{\base{B}} p
    \]
\end{Lemma}

Rather than focusing on the notion of derivability, we may recover a version of the generalization by focusing of \emph{valid argument}. This is the main result of the paper, Theorem~\ref{thm:ptv-bes-equiv}:
\[
\Gamma \supp_\base{B} \phi \qquad \mbox{iff} \qquad  \Gamma \entails_\base{B} \phi
\]
Since P-tV considers arguments and their dynamics and B-eS represents only a judgment of formulae, we may think of them as the operational and declarative counterparts of the same proof-theoretic semantics.

\section{From Proof-theoretic Validity to Base-extension Semantics} \label{sec:main}

\subsection{Constructivity, Intuitionism, and Validity}
Schroeder-Heister~\cite{Schroeder2007modelvsproof} observes that P-tV represents a `constructive' notion of validity. In this section, we explore this idea as it explicates how P-tV relates to B-eS. We can understand this through \emph{realizability}: formulas are \emph{realized} by objects, known as \emph{realizers}, whose existence certifies the \emph{validity} of the formula.

In the BHK interpretation of IPL, a realizer of a formula $\phi \to \psi$ is a function that takes as inputs a realizer of $\phi$ and outputs a realizer of $\psi$ --- see Section~\ref{sec:bhk}. This functional relation is what makes the semantics constructive, as the functions construct a realizer for $\psi$ from a realizer for $\phi$.

In P-tV, a realizer for $\phi \to \psi$ is a valid argument for $\phi \to \psi$. This is identified with a valid argument \emph{from $\phi$ to $\psi$}. Thus, in adopting P-tV, we move from functions to arguments as the method for creating realizers.

Recall that P-tV (Definition~\ref{def:bvalid}) is paramaterized on a sets of reductions $\mathbb{R}$. Precisely what notion of constructivity it expresses depends on what reductions are adopted. Presently, we present some conditions that should express the intuitionistic notion of constructivity. In Section~\ref{sec:conclusion}, we discuss how this analysis changes as the target logic changes. 

Firstly, for intuitionism, we expect at least the introduction rules to be valid constructors. Thus:

\begin{Condition} \label{cond:intros}
The set of reductions $\mathbb{R}$ is such that uses of the introduction rules comprise valid arguments.
\end{Condition}

For example, the argument
\[
\infer{\phi_1 \land \phi_2}{\phi_1 & \phi_2}
\]
is a valid argument for any $\phi$. By Definition~\ref{def:bvalid}, this means that for any $\base{B}$-valid arguments for $\phi$ and $\psi$, one has a $\base{B}$-valid argument for $\phi \land \psi$.

Observe that the traditional reduction operators by Prawitz~\cite{Prawitz2006natural} suffice for Condition~\ref{cond:intros}. This follows as Definition~\ref{def:bvalid} essentially requires one to assume a $\base{B}$-valid argument for the premises and then apply an elimination rule, and the reductions precisely return the assumed argument.

Secondly, we require understanding the conditions under which one has a realizer for $\bot$. In contrast to BHK, the answer ought not to be `never' in P-tV, as presumably one should be able to conclude $\bot$ from an inconsistent set of assumptions. We return to this in Section~\ref{sec:efq}. Dummett~\cite{Dummett1991logical} proposes the following for intuitionism:

\begin{Condition} \label{cond:bot}
The set of reductions $\mathbb{R}$ is such that the following holds:
\[
\Gamma \entails_{\base{B}} \bot \qquad \text{iff} \qquad \Gamma \entails_{\base{B}} p \text{ for all } p \in \setAtoms
\]
\end{Condition}

% Recall that P-tV proceeds through a relative validity condition. Thus, more precisely, a $\base{B}$-realizer for $\phi \to \psi$ is a $\base{B}$-valid argument such that there is a $\base{B}$-valid argument from $\phi$ to $\psi$ (Definition~\ref{def:bvalid} and Condition~\ref{cond:intros}).

Thirdly, Schroeder-Heister~\cite{Schroeder2012categorical} observes that constructive semantics follows the standard dogma of semantics in which the validity of a consequence is understood as the \emph{transmission} of the validity of assumptions to the conclusion. Hence, by Definition~\ref{def:bvalid}\ref{def:bvalid:three}, a $\base{B}$-valid argument from $\phi$ to $\psi$ transmits $\base{C}$-validity of $\phi$ to $\psi$ for $\base{C} \supseteq \base{B}$. Hence, accepting this standard dogma, we impose the following:

\begin{Condition}\label{cond:inf}
The set of reductions $\mathbb{R}$ is such that the following holds:
\[
\Gamma \entails_{\base{B}} \phi \qquad \text{iff} \qquad \forall \base{C} \supseteq \base{B}, \text{ if } \entails_{\base{C}} \psi \text{ for } \psi \in \Gamma, \text{ then } \entails_{\base{C}} \phi
\]
\end{Condition}

This concludes the analysis of constructiveness, intuitionism, and validity.

\subsection{From Proof-theoretic Validity to Support} \label{Chapter:PtV_BeS:Sat_Supp}
In summary, the conditions for reductions \(\mathbb{R}\) are such that for any base \(\base{B}\):
\begin{itemize}[label=--]
    \item (Condition~\ref{cond:disj-elim}) The full disjunction-elimination rule applies, not only when it has atomic conclusions.
    \item (Condition~\ref{cond:atomic}) If there is a \(\base{B}\)-valid argument from atomic assumptions to an atomic conclusion, then there is a \(\base{B}\)-derivation from those assumptions to the conclusion.
    \item (Condition~\ref{cond:comp}) A formula is not required as an assumption if there is a \(\base{B}\)-valid argument with the same open assumptions except the formula.
    \item (Condition~\ref{cond:intros}) The introduction rules construct \(\base{B}\)-valid arguments.
    \item (Condition~\ref{cond:bot}) There is a \(\base{B}\)-valid argument concluding \(\bot\) iff there are \(\base{B}\)-valid arguments concluding \(p\) for any atom \(p \in \setAtoms\).
    \item (Condition~\ref{cond:inf}) The standard dogma of semantics --- that is, the transmission view of consequence --- applies to \(\base{B}\)-validity.
\end{itemize}

We call a set of reductions satisfying these conditions \emph{supportive}. Relative to such reductions, entailment (\(\text{P-tV}\)) and support (\(\text{B-eS}\)) coincide. This is the content of Theorem~\ref{thm:ptv-bes-equiv}, below.

\begin{Lemma} \label{lem:reds}
If the set of reductions is supportive, the following hold:
\begin{align}
\Gamma \entails_\base{B} \phi_1 \land \phi_2 &\quad \text{iff} \quad \Gamma \entails_\base{B} \phi_1 \text{ and } \Gamma \entails_\base{B} \phi_2 \tag{$\land$-cond} \label{cond:conj}\\
\chi_1 \lor \chi_2, \Gamma \entails_\base{B} \phi \hspace{9.5mm} &\quad \text{iff} \quad \chi_1, \Gamma \entails_\base{B} \phi \text{ and } \chi_2, \Gamma \entails_\base{B} \phi \tag{$\lor$-cond} \label{cond:disj} \\
 \Gamma \entails_\base{B} \phi \to \psi \hspace{2mm} &\quad \text{iff} \quad \phi, \Gamma \entails_\base{B} \psi \tag{$\to$-cond} \label{cond:imp}
\end{align}
\end{Lemma}

We demonstrate only one illustrative case, the others being similar.

\begin{proof}[Proof of \ref{cond:disj}]
First, we show that \(\chi_1 \lor \chi_2, \Gamma \entails_\base{B} \phi\) implies \(\chi_1, \Gamma \entails_\base{B} \phi\) and \(\chi_2, \Gamma \entails_\base{B} \phi\). Let \(\argument{A}\) be a \(\base{B}\)-valid argument witnessing \(\chi_1 \lor \chi_2, \Gamma \entails_\base{B} \phi\). If \(\chi_1 \lor \chi_2\) does not occur in \(\argument{A}\), then the proposition holds trivially. If \(\chi_1 \lor \chi_2\) does occur in \(\argument{A}\), consider the following argument \(\argument{C}_i\) for \(i \in \{1,2\}\):
\[
\ddeduce{\phi}{\argument{A}}{\ddeduce{\chi_1 \lor \chi_2}{\quad \argument{B}_i}{\quad \chi_i}, \Gamma}
\]
where \(\argument{B}_i\) is simply the use of \(\irn \lor\). Observe that \(\argument{B}_i\) is \(\base{B}\)-valid because of Condition~\ref{cond:intros}. Hence, since \(\argument{A}\) and \(\argument{B}\) are \(\base{B}\)-valid, by \ref{cond:comp}, we infer that \(\argument{C}_i\) is \(\base{B}\)-valid for \(i \in \{1,2\}\).

Second, we show that \(\chi_1, \Gamma \entails_\base{B} \phi\) and \(\chi_2, \Gamma \entails_\base{B} \phi\), together imply \(\chi_1 \lor \chi_2, \Gamma \entails_\base{B} \phi\). Let $\argument{A}_1$ and $\argument{A}_2$ be valid arguments for the assumptions. Let $\argument{B}$ be as follows:
\[
\infer{\phi}{\chi_1\lor \chi_2 & \ddeduce{\phi}{\argument{A}_1}{[\chi_1],\Gamma}
&
\ddeduce{\phi}{\argument{A}_2}{[\chi_2],\Gamma}
}
\]
By Condition~\ref{cond:disj-elim}, argument $\argument{B}$ is $\baseB$-valid argument. It witnesses the desired conclusion. 
\end{proof}

This Lemma simplifies the presentation of the argument that proof-theoretic validity is equivalent to support; that is, the proof of Theorem~\ref{thm:ptv-bes-equiv}.

\begin{Proposition}[Monotonicity in P-tV] \label{prop:monotonicity-ptv}
If \(\Gamma \entails_{\base{B}} \chi\) and \(\base{C} \supseteq \base{B}\), then \(\Gamma \entails_{\base{C}} \chi\).
\end{Proposition}
\begin{proof}
This follows immediately from Definition~\ref{def:bvalid} by the monotonicity of derivability in a base --- that is, \(\proves_{\base{B}} \at{p}\) implies \(\proves_{\base{C}} \at{p}\) for any \(\base{C} \supseteq \base{B}\).
\end{proof}

\begin{Proposition}
[Monotonicity in B-eS] \label{prop:monotonicity-bes}
If \(\Gamma \supp_{\base{B}} \chi\) and \(\base{C} \supseteq \base{B}\), then \(\Gamma \supp_{\base{C}} \chi\).
\end{Proposition}
\begin{proof}
This was shown by Sandqvist~\cite{Sandqvist2015IL}.
\end{proof}

We regard these propositions as sufficiently basic statements that we may apply them without explicit reference.

\begin{Theorem}\label{thm:ptv-bes-equiv} 
Assuming the set of reductions for arguments is supportive, 
\[
\mbox{
$\Gamma \entails_{\base{B}} \phi$ \qquad iff \qquad  $\Gamma \supp_{\base{B}} \phi$ 
}
\]
\end{Theorem}
\begin{proof}
 We proceed by induction on the multiset ordering induced by the ordering on the logical size of formulas (i.e., the number of logical constants they contain):
    \begin{itemize}[label=--]
        \item \textsc{Base Case.} We have $\Gamma \cup \{\phi\}  \subseteq \setAtoms$. We reason as follows:
        \begin{align}
            \Gamma \entails_{\base{B}} \phi & \qquad \mbox{iff} \qquad \Gamma \proves_{\base{B}} \phi \tag{Cond.~\ref{cond:atomic}} \\
            & \qquad \mbox{iff} \qquad \Gamma \supp_{\base{B}} \phi \tag{Lemma~\ref{lem:sand-base}} 
        \end{align}
        \item \textsc{Inductive Step.} There is a non-atomic $\chi \in \Gamma\cup \{\phi\}$. We distinguish case when $\chi$ is on the right and on the left, $\chi = \phi$ and $\chi \in \Gamma$. Certain steps are immediate consequence of unpacking Definition~\ref{def:supp}, they have been labelled by the definition without further elaboration. Throughout, we label the use of the \emph{induction hypothesis} by `IH'. \smallskip

        \noindent Let $\chi = \phi$. We proceed by case analysis on the structure of $\phi$:
        \begin{itemize}[label=--]
            \item $\phi = \bot$. We reason as follows:
            \begin{align}
                \Gamma \entails_{\base{B}} \bot & \qquad \mbox{iff} \qquad \mbox{$\Gamma \entails_{\base{B}} p$ for all $p \in \setAtoms$} \tag{Cond.~\ref{cond:bot}} \\ & \qquad \mbox{iff} \qquad 
                \mbox{$\Gamma \supp_{\base{B}} p$ for all $p \in \setAtoms$} \tag{IH} \\
                & \qquad \mbox{iff} \qquad 
 \mbox{$\Gamma \supp_{\base{B}} \bot$} \tag{Def.~\ref{def:supp}} 
            \end{align}
            
            \item $\phi = \phi_1 \land \phi_2$. We reason as follows:
           \begin{align}
               \Gamma \entails_{\base{B}} \phi_1\land\phi_2  & \qquad \mbox{iff} \qquad \Gamma \entails_{\base{B}} \phi_1 \mbox{ and } \Gamma \entails_{\base{B}} \phi_2 \tag{Lemma~\ref{lem:reds}} \\
                 & \qquad \mbox{iff} \qquad  \Gamma \supp_{\base{B}} \phi_1 \mbox{ and } \Gamma \supp_{\base{B}} \phi_2 \tag{IH}\\
                 & \qquad \mbox{iff} \qquad  \Gamma \supp_{\base{B}} \phi_1 \land \phi_2 \tag{Def.~\ref{def:supp}}
            \end{align}

            \item $\phi = \phi_1 \lor \phi_2$.  We reason as follows:
              \begin{align}
               \Gamma \entails_{\base{B}} \phi_1\lor \phi_2  & \quad \mbox{iff} \quad \mbox{$\forall \base{C} \supseteq \base{B}\, \forall p \in \setAtoms$,} \notag \\ & \hspace{3em} \mbox{ if $\Gamma, \phi_1 \lor \phi_2 \entails_{\base{B}} p$, then $\Gamma \entails_{\base{B}} p$} \tag{Cond.~\ref{cond:comp}} \\
               & \quad \mbox{iff} \quad \mbox{$\forall \base{C} \supseteq \base{B}\, \forall p \in \setAtoms$,} \notag \\ & \hspace{3em} \mbox{ if $\Gamma, \phi_1 \entails_{\base{B}} p$ and $\Gamma, \phi_2 \entails_{\base{B}} p$, then $\Gamma \entails_{\base{B}} p$} \tag{Lemma~\ref{lem:reds}} \\
                 & \quad \mbox{iff} \quad \mbox{$\forall \base{C} \supseteq \base{B}\, \forall p \in \setAtoms$,} \notag \\ & \hspace{3em} \mbox{ if $\Gamma, \phi_1  \supp_{\base{B}} p$ and $\Gamma, \phi_2 \supp_{\base{B}} p$, then $\Gamma \supp_{\base{B}} p$} \tag{IH} \\
                 & \quad \mbox{iff} \quad \mbox{$\Gamma \supp_{\base{B}} \phi_1 \lor \phi_2$} \tag{Def.~\ref{def:supp}}
            \end{align}
        \item $\chi = \phi_1 \to \phi_2$. We reason as follows:
        \begin{align}
               \Gamma \entails_{\base{B}} \phi_1\to \phi_2  & \quad \mbox{iff} \quad  \Gamma, \phi_1 \entails_{\base{B}}  \phi_2 \tag{Lemma~\ref{lem:reds}} \\
               & \quad \mbox{iff} \quad
               \mbox{$\forall \base{C} \supseteq \base{B}$, if $\entails_{\base{C}} \psi$ for $\psi \in \Gamma, \phi_1$, then $\entails_{\base{C}} \phi_2$} \tag{Cond.~\ref{cond:comp}} \\
                 & \quad \mbox{iff} \quad
               \mbox{$\forall \base{C} \supseteq \base{B}$, if $\supp_{\base{C}} \psi$ for $\psi \in \Gamma, \phi_1$, then $\supp_{\base{C}} \phi_2$}  \tag{IH} \\
                 & \quad \mbox{iff} \quad \mbox{$\Gamma \supp_{\base{B}} \phi_1 \to \phi_2$} \tag{Def.~\ref{def:supp}}
            \end{align}
        \end{itemize}
        
        \noindent This completes the case analysis. \smallskip

       \noindent Let $\chi \in \Gamma$. That is, we have $\chi,\Delta \entails_{\base{B}} \phi$ for some set $\Delta$. We proceed by case analysis on the structure of $\chi$:
        
        \begin{itemize}[label=--]
             \item $\chi = \chi_1 \land \chi_2$. We reason as follows:
   \begin{align}
               \chi_1 \land \chi_2, \Delta \entails_{\base{B}} \phi  & \mbox{ iff }  \mbox{$\forall \base{C} \supseteq \base{B}$, if $\entails_{\base{C}} \psi$ for $\psi \in \Delta$ and $\entails_{\base{C}} \chi_1 \land \chi_2$, then $ \entails_{\base{B}} \phi$} \tag{Lemma~\ref{lem:reds}} \\
               & \mbox{ iff } \mbox{$\forall \base{C} \supseteq \base{B}$, if $\entails_{\base{C}} \psi$ for $\psi \in \Delta \cup \{\chi_1,\chi_2\}$, then $ \entails_{\base{B}} \phi$} \tag{Lemma~\ref{lem:reds}} \\
                 & \mbox{ iff } \mbox{$\forall \base{C} \supseteq \base{B}$, if $\supp_{\base{C}} \psi$ for $\psi \in \Delta \cup \{\chi_1,\chi_2\}$, then $\supp_{\base{B}} \phi$} \tag{IH} \\
                 & \mbox{ iff } \mbox{$\forall \base{C} \supseteq \base{B}$, if $\supp_{\base{C}} \psi$ for $\psi \in \Delta$ and $ \supp_{\base{C}} \chi_1 \land \chi_2$, then $\supp_{\base{B}} \phi$} \tag{Def.~\ref{def:supp}} \\
                 & \mbox{ iff } \mbox{$\chi_1\to\chi_2,\Delta \supp_{\base{B}} \phi$} \tag{Def.~\ref{def:supp}}
            \end{align}
             \item $\chi = \chi_1 \lor \chi_2$.  We reason as follows:
                \begin{align}
               \chi_1 \lor \chi_2, \Delta \entails_{\base{B}} \phi  & \quad \mbox{iff} \quad  \chi_1,\Delta \entails_{\base{B}} \phi \mbox{ and } \chi_2,\Delta \entails_{\base{B}} \phi  \tag{Lemma~\ref{lem:reds}} \\
               & \quad \mbox{iff} \quad  \chi_1,\Delta \supp_{\base{B}} \phi \mbox{ and } \chi_2,\Delta \supp_{\base{B}} \phi   \tag{IH} \\
                 & \quad \mbox{iff} \quad  \chi_1 \lor \chi_2,\Delta \supp_{\base{B}} \phi \tag{Def.~\ref{def:supp}}
            \end{align}
            \item $\chi  = \chi_1 \to \chi_2$. We reason as follows:
           \begin{align}
               \chi_1 \to \chi_2, \Delta \entails_{\base{B}} \phi  & \mbox{ iff }  \mbox{$\forall \base{C} \supseteq \base{B}$, if $\entails_{\base{C}} \psi$ for $\psi \in \Delta$ and $\entails_{\base{C}} \chi_1 \to \chi_2$, then $ \entails_{\base{B}} \phi$} \tag{Lemma~\ref{lem:reds}} \\
               & \mbox{ iff } \mbox{$\forall \base{C} \supseteq \base{B}$, if $\entails_{\base{C}} \psi$ for $\psi \in \Delta$ and $\chi_1 \entails_{\base{C}} \chi_2$, then $ \entails_{\base{B}} \phi$} \tag{Lemma~\ref{lem:reds}} \\
                 & \mbox{ iff } \mbox{$\forall \base{C} \supseteq \base{B}$, if $\supp_{\base{C}} \psi$ for $\psi \in \Delta$ and $\chi_1 \supp_{\base{C}} \chi_2$, then $\supp_{\base{B}} \phi$} \tag{IH} \\
                 & \mbox{ iff } \mbox{$\forall \base{C} \supseteq \base{B}$, if $\supp_{\base{C}} \psi$ for $\psi \in \Delta$ and $ \supp_{\base{C}} \chi_1 \to \chi_2$, then $\supp_{\base{B}} \phi$} \tag{Def.~\ref{def:supp}} \\
                 & \mbox{ iff } \mbox{$\chi_1\to\chi_2,\Delta \supp_{\base{B}} \phi$} \tag{Def.~\ref{def:supp}}
            \end{align}
    \end{itemize}
    \end{itemize}
    This completes the induction.
\end{proof}

        %     We proceed by case analysis on the structure of $\chi_1$:
        % \begin{itemize}
        %     \item[-] $\chi_1 \in \setAtoms \cup \{\bot\}$. By Proposition~\ref{prop:weakening}, 
        %     $\chi_1,\chi_1 \to \chi_2,\Delta \supp_{\base{B}} \phi$
        %     . Hence, by Proposition~\ref{prop:implications}, $
        %     \chi_1,\chi_2,\Delta \supp_{\base{B}} \phi$. By the IH, $
        %     \chi_1,\chi_2,\Delta \supp_{\base{B}} \phi$. It follows that $\chi_1 \to \chi_2,\Delta \supp_{\base{B}} \phi$. 
        %     \item[-] $\chi_1 = \psi_1 \land \psi_2$. By Proposition~\ref{prop:implications},  $\psi_1 \to (\psi_2 \to \chi_2),\Delta \supp_{\base{B}} \phi$. By the IH, $\psi_1 \to (\psi_2 \to \chi_2),\Delta \supp_{\base{B}} \phi$. It follows that $(\psi_1 \land \psi_2) \to \chi_2,\Delta \supp_{\base{B}} \phi$.
        %     \item[-] $\chi_1 = \psi_1 \lor \psi_2$. By Proposition~\ref{prop:implications},  $(\psi_1 \to \chi_2),(\psi_2 \to \chi_2),\Delta \supp_{\base{B}} \phi$. By the IH, $(\psi_1 \to \chi_2),(\psi_2 \to \chi_2),\Delta \supp_{\base{B}} \phi$.  It follows that  $((\psi_1 \lor \psi_2) \to \chi_2,\Delta \supp_{\base{B}} \phi$.
        %     \item[-] $\chi_1 = \psi_1 \to \psi_2$. By Proposition~\ref{prop:specialimplication},  $\psi_1 \supp_{\base{B}} \psi_2$ implies $\chi_2,\Delta \supp_{\base{B}} \phi$. Therefore, by Theorem~\ref{thm:Sandqvist} and the IH, $\psi_1 \supp_{\base{B}} \psi_2$ implies $\chi_2,\Delta \supp_{\base{B}} \phi$. It follows that $(\psi_1 \to \psi_2) \to \chi_2,\Delta \supp_{\base{B}} \phi$.
        % \end{itemize}
        % This completes the case analysis for $\chi_1$.
        % \end{itemize}

A corollary is an affirmative answer to Prawitz's Conjecture for P-tV based on the elimination rules:

\begin{Corollary} \label{cor:completeness}
   Assuming the set of reductions for arguments is supportive and restricting to Sandqvist bases,
    \[
    \Gamma \entails \phi \qquad \mbox{iff} \qquad \Gamma \proves \phi
    \]
\end{Corollary}
\begin{proof}
We reason as follows:
\begin{align}
    \Gamma \entails \phi & \qquad \mbox{iff} \qquad  \mbox{$\Gamma \entails_{\base{B}} \phi$ for all $\base{B} \in \mathfrak{S}$.} \tag{Def.~\ref{def:bvalid}} \\
     & \qquad \mbox{iff} \qquad   \mbox{$\Gamma \supp_{\base{B}} \phi$ for all $\base{B} \in \mathfrak{S}$.} \tag{Theorem.~\ref{thm:ptv-bes-equiv}} \\
      & \qquad \mbox{iff} \qquad   \mbox{$\Gamma \supp \phi$} \tag{Def.~\ref{def:supp}}\\
      & \qquad \mbox{iff} \qquad   \mbox{$\Gamma \proves \phi$} \tag{Theorem~\ref{thm:sandqvist:snc}}
\end{align}
\end{proof}

We have given precise conditions under which P-tV and B-eS coincide. It remains to determine precise what sets of reductions are indeed supportive; as Schroeder-Heister~\cite{Schroeder2015elim} observes, the reduction used by Prawitz~\cite{Prawitz2006natural} do not suffice (see Section~\ref{sec:ptv}). We leave this to future work. 

\section{Ex Falso Quodlibet?} \label{sec:efq}

 In Section~\ref{sec:main}, we accepted the meaning of $\bot$ in terms of \emph{ex falso quodlibet} (\rn{EFQ}),
\[
\infer{\phi}{\bot}
\]
The section is motivated by the BHK interpretation of intuitionism (Section~\ref{sec:bhk}) in which
\[
\mbox{`nothing is a proof of $\bot$'} 
\] 
This causes an apparent conflict in this paper that requires some explanation. 

While realizability and proof-theoretic validity are deeply connected, they are not the same thing. The realizability interpretation takes place at an essentially classical meta-level, while proof-theoretic validity takes place at an essentially intuitionistic meta-level.  

What we mean when we say that realizability is classical is that this `if\ldots, then \ldots' in its clauses are classical. In the parlance of realizability, $\rn{EFQ}$ says the following: If there is a realizer for $\bot$, then there is a realizer for $\phi$.  Since the antecedent is false, the implication hold vacuously.

By contrast, when we say that proof-theoretic validity is intuitionistic, we mean that we have chosen the notion of validity to be such that: given a $\base{B}$-valid argument for $\bot$, one may construct a $\base{B}$-valid argument for $\phi$. For this condition to be a \emph{definition}, we apply a closed-world assumption in the form of \emph{definitional reflection} (DR). In this case, the expression of DR is somewhat different from the version given in Section~\ref{sec:intro}. Schroeder-Heister~\cite{Schroeder2015elim} has given a detailed account and observes that for the particular version of DR we use, our bias is  `to consider ``consequential'' clauses'  as definitions, rather than introduction. This is captured in Condition~\ref{cond:bot} (Section~\ref{sec:main}), which establishes \rn{EFQ} as the definition of $\bot$.

This distinction between realizability and proof-theoretic validity is why IPL is not `structurally' complete --- see, for example, Pogorzelski~\cite{Pogorzelski}. That is, while the horizontal bar in natural deduction corresponds to realizablity --- that is, the existence of realizers for the things above it guarantee the existence of the things below it --- the implication corresponds to proof-theoretic validity --- that is, there is a $\base{B}$-valid argument for $\phi \to \psi$ iff there is a $\base{B}$-valid argument from $\phi$ to $\psi$.

To end this section, we note that the remarks about the constructiveness of \rn{EFQ} being classically justified is not new.  Indeed, it is an essential part of the standard analysis on the relationship between \rn{EFQ} and the \emph{disjunctive syllogism} (\rn{DS}),
\[
\infer{\psi}{\phi \lor \psi & \neg \phi}
\]
--- see, for example, Johansson and Heyting \cite{derMolen} and Pereira et al.~\cite{Pereira2024}. How does the existence of valid arguments for $\phi \lor \psi$ and $\neg \phi$ neccesitate the existence of a valid argument for $\psi$? It is not that one is constructed from them, rather it is that in this situation, on the basis of the classical reasoning at the meta-level in which the arguments exist, we conclude that there must be a valid argument for $\psi$. The details are as described below.

Suppose one has $\base{B}$-valid arguments $\mathcal{D}_1$ and $\mathcal{D}_2$ for $\phi\lor \psi$ and $\neg \phi$, respectively. The existence of $\mathcal{D}_1$ suggests that there is a  $\base{B}$-valid argument $\mathcal{D}_1'$ for either $\phi$ or $\psi$. Suppose that it is $\mathcal{D}_1'$ is a $\base{B}$-valid argument for $\phi$, then using $\mathcal{D}_2$, we can construct a $\base{B}$-valid argument for $\bot$ using $\ern \to$,
\[
\infer[\ern \to]{\bot}{[\phi] & \deduce{\neg \phi}{\mathcal{D}_2}}
\]
However, there is no $\base{B}$-valid argument for $\bot$. Contradiction! Hence, we must reject our assumption, and the only remaining possibility is that $\mathcal{D}_1'$ is a $\base{B}$-valid argument for $\psi$ (not $\phi$). We thus conclude, by classical reasoning, from the existence of $\base{B}$-valid argument for $\phi \lor \psi$ and $\neg \phi$ that there is a $\base{B}$-valid argument for $\psi$ without constructing one.

This work on \rn{EFQ} admubrates the readings of $\bot$ by  Tennant~\cite{Tennant78entailment,Tennant2017core} and Fukuda and Igarashi~\cite{fukuda} in which it is a declaration about the \emph{state} of a construction. More generally, Berto~\cite{berto} has developed a reading of negation as a `modal' operator making a declaration about the set of accessible worlds.

\section{Conclusion}
\label{sec:conclusion}

Proof-theoretic semantics is the approach to meaning based on \emph{proof} (as opposed to \emph{truth}). There are two broad approaches to it in the literature: proof-theoretic validity (P-tV) and base-extension semantics (B-eS). The former is a semantics of arguments, and the latter is a semantics of a logic \emph{in terms of} arguments. Heuristically, P-tV provides a semantics by taking a sequent as valid iff it admits a valid argument. In this paper, we demonstrate that a certain version of P-tV provided by Prawitz~\cite{Prawitz1971ideas} (see also Schroeder-Heister~\cite{Schroeder2015elim}) contains the same semantic content as the B-eS for IPL provided by Sandqvist
\cite{Sandqvist2015IL}. This explains why this B-eS is complete.

To make the connection between P-tV and the B-eS of IPL, the paper considers carefully the notions of \emph{reduction} and \emph{base} that capture the `constructive' content of intuitionistic proof. This follows from the BHK interpretation of intuitionism (see Section~\ref{sec:background}). This, of course, raises the question of whether or no the other logics can be similarly captured. For example, there are B-eS for a variety of modal~\cite{Eckhardt} and intuitionistic substructural logics~\cite{IMLL,BI} relative to which the kind of analysis in this paper could be performed. Thus, future work includes extending the analysis herein to these domains and thereby understanding the consequential reading and constructive content of these logics.  

On this note, we may particularly ask for a truly consequential reading of \emph{classical} entailment. According to Dummett~\cite{Dummett1978-DUMTJO}:
\begin{quotation}
    In the resolution of the conflict between  these two views [the truth-theoretic reading of classical connectives, and the demand that it be explained without recourse to the principle of \emph{bivalence}] lies, as I see it, one of the most fundamental and intractable problems in the theory of meaning; indeed in all philosophy. 
\end{quotation}
Sandqvist~\cite{Sandqvist2009CL} addressed this with a B-eS akin to the one in this paper. This work takes $\to$ and $\bot$ as the only primitive connectives and provides the same B-eS for IPL, but with atomic systems restricted to rules without discharge (other choices are also possible — see Sandqvist~\cite{sandqvistwld,sandqvist2015hypothesis}). This provided an anchor relative to which one can investigate classical logic.

In parallel, Gheorghiu and Pym~\cite{gheorghiu2023constraints} have shown that the key factor driving the proof-theoretic semantic distinction between intuitionistic and classical logic lies in the interpretation of disjunction (cf. Dummett~\cite{dummett2000elements}). Moreover, just as intuitionistic logic corresponds to the (simply typed) $\lambda$-calculus as a canonical instantiation of the realizability (i.e., BHK) interpretation, classical logic corresponds to the $\lambda\mu$-calculus (see Parigot~\cite{Parigot2005}). In this context, Pym and Ritter~\cite{PR2001disj} have shown that one can give two natural interpretations of disjunction, both of which are constructive, through $\lambda\mu$-terms, one corresponding to intuitionistic disjunction and the other corresponding to classical disjunction. Investigating the concept of proof-theoretic validity for classical logic relative to these findings remains future work.

\subsection*{Acknowledgments}

We are grateful to the reviewers on an earlier edition of this article for their thorough and thoughtful comments on this work.

\subsection*{Declaration}

This work has been partially supported by the UK EPSRC grants EP/S013008/1 and EP/R006865/1, and by Gheorghiu's EPSRC Doctoral Studentship.

\bibliographystyle{siam}
\bibliography{bib}

\end{document}

%% file: macro.tex
\renewcommand{\emptyset}{\varnothing}
\renewcommand{\phi}{\varphi}

\newcommand{\argument}[1]{\ensuremath{\mathcal{#1}}}

\newcommand{\base}[1]{\ensuremath{\mathscr{#1}}}
\newcommand{\baseB}{\base{B}}

\newcommand{\calculus}[1]{\ensuremath{\mathsf{#1}}}
\newcommand{\system}[1]{\ensuremath{\mathsf{#1}}}
\newcommand{\set}[1]{\ensuremath{\mathbb{#1}}}
\newcommand{\setAtoms}{\set{A}}
\newcommand{\setFormulas}{\set{F}}
\newcommand{\setSubs}{\set{S}}

\newcommand{\proves}{\vdash}
\newcommand{\entails}{\vDash}
\newcommand{\supp}{\Vdash}
\newcommand{\seq}{\mathbin{\triangleright}}

\newcommand{\at}[1]{#1}

\newcommand{\redtostar}{\rightsquigarrow^*}
\newcommand{\redto}{\rightsquigarrow}

\newcommand{\llbracket}{[\![}
\newcommand{\rrbracket}{]\!]}

\newcommand*{\rn}[1]  {\ensuremath{\mathsf{#1}}}

\newcommand*{\irn}[2][]  {\rn{#2_{I#1}}}
\newcommand*{\ern}[2][]  {\rn{#2_{E#1}}}

\newcommand{\ddeduce}[3]{\deduce[#2]{#1}{#3}}

\makeatletter
\def\labelandtag#1#2{\begingroup
   \def\@currentlabel{#2}%
   \phantomsection\label{#1}\endgroup
}
\makeatother

%% file: main.bbl
\begin{thebibliography}{10}

\bibitem{Barendregt1991}
{\sc H.~Barendregt}, {\em {Introduction to Generalized Type Systems}}, Journal
  of Functional Programming, 1 (1991), pp.~125--154.

\bibitem{BDS2013}
{\sc H.~Barendregt, W.~Dekkers, and R.~Statman}, {\em {Lambda Calculus with
  Types}}, Perspectives in Logic, Cambridge University Press, 2013.

\bibitem{Barendregt1993}
{\sc H.~P. Barendregt}, {\em {Lambda calculi with types}}, Oxford University
  Press, Inc., USA, 1993, p.~117–309.

\bibitem{berto}
{\sc F.~Berto}, {\em {A Modality Called `Negation'}}, Mind, 124 (2015),
  pp.~761--793.

\bibitem{Brandom2000}
{\sc R.~Brandom}, {\em {Articulating Reasons: An Introduction to
  Inferentialism}}, Harvard University Press, 2000.

\bibitem{brouwer1913intuitionism}
{\sc L.~E.~J. Brouwer}, {\em {Intuitionism and Formalism}}, Bulletin of the
  American Mathematical Society, 20 (1913), pp.~81--96.

\bibitem{cartmell}
{\sc J.~Cartmell}, {\em {Generalised Algebraic Theories and Contextual
  Categories}}, PhD thesis, Oxford University, 1978.

\bibitem{derMolen}
{\sc T.~v. der Molen}, {\em {The Johansson/Heyting Letters and the Birth of
  Minimal Logic}}.
\newblock \url{https://eprints.illc.uva.nl/id/eprint/696/1/X-2016-04.text.pdf},
  2016.
\newblock Accessed May 2024.

\bibitem{Dummett1978-DUMTJO}
{\sc M.~Dummett}, {\em The justification of deduction}, in Truth and other
  Enigmas, M.~Dummett, ed., Duckworth \& Co, 1978.

\bibitem{Dummett1991logical}
{\sc M.~Dummett}, {\em {The Logical Basis of Metaphysics}}, Harvard University
  Press, 1991.

\bibitem{dummett2000elements}
\leavevmode\vrule height 2pt depth -1.6pt width 23pt, {\em {Elements of
  Intuitionism}}, vol.~39, Oxford University Press, 2000.

\bibitem{Eckhardt}
{\sc T.~Eckhardt and D.~J. Pym}, {\em {Proof-theoretic Semantics for Modal
  Logics}}, Logic Journal of the IGPL,  (2024).
\newblock accepted.

\bibitem{Ferreira2006}
{\sc F.~Ferreira}, {\em {Comments on Predicative Logic}}, Journal of
  Philosophical Logic, 35 (2006), pp.~1--8.

\bibitem{Ferreria2013}
{\sc F.~Ferreira and G.~Ferreira}, {\em {Atomic Polymorphism}}, The Journal of
  Symbolic Logic, 78 (2013), pp.~260--274.

\bibitem{ferreira2014faithfulness}
{\sc F.~Ferreira and G.~Ferreira}, {\em {The faithfulness of atomic
  polymorphism}}, Trends in Logic XIII,  (2014).

\bibitem{francez2015proof}
{\sc N.~Francez}, {\em {Proof-theoretic semantics}}, vol.~573, College
  Publications London, 2015.

\bibitem{fukuda}
{\sc Y.~Fukuda and R.~Igarashi}, {\em {A Reconstruction of Ex Falso Quodlibet
  via Quasi-Multiple-Conclusion Natural Deduction}}, in Logic, Rationality, and
  Interaction, A.~Baltag, J.~Seligman, and T.~Yamada, eds., Springer, 2017,
  pp.~554--569.

\bibitem{IMLL}
{\sc A.~V. Gheorghiu, T.~Gu, and D.~J. Pym}, {\em {Proof-theoretic Semantics
  for Intuitionistic Multiplicative Linear Logic}}, in Automated Reasoning with
  Analytic Tableaux and Related Methods --- TABLEAUX, R.~Ramanayake and
  J.~Urban, eds., Springer, 2023, pp.~367--385.

\bibitem{BI}
\leavevmode\vrule height 2pt depth -1.6pt width 23pt, {\em {Proof-theoretic
  Semantics for the Logic of Bunched Implications}}, arXiv:2311.16719,  (2023).
\newblock Accessed February 2024.

\bibitem{NAF}
{\sc A.~V. Gheorghiu and D.~J. Pym}, {\em { Definite Formulae,
  Negation-as-Failure, and the Base-extension Semantics for Intuitionistic
  Propositional Logic }}, Bulletin of the Section of Logic,  (2023).

\bibitem{gheorghiu2023constraints}
\leavevmode\vrule height 2pt depth -1.6pt width 23pt, {\em {Defining Logical
  Systems via Algebraic Constraints on Proofs}}, Journal of Logic and
  Computation,  (2023).

\bibitem{goldfarb2016dummett}
{\sc W.~Goldfarb}, {\em {On Dummett's ``Proof-theoretic justifications of
  logical laws''}}, in Advances in Proof-theoretic Semantics, Springer, 2016,
  pp.~195--210.

\bibitem{hallnas2006proof}
{\sc L.~Halln{\"a}s}, {\em {On the Proof-theoretic Foundation of General
  Definition Theory}}, Synthese, 148 (2006), pp.~589--602.

\bibitem{hallnas1990proof}
{\sc L.~Halln{\"a}s and P.~Schroeder-Heister}, {\em {A Proof-theoretic Approach
  to Logic Programming: I. Clauses as Rules}}, Journal of Logic and
  Computation, 1 (1990), pp.~261--283.

\bibitem{hallnas1991proof}
\leavevmode\vrule height 2pt depth -1.6pt width 23pt, {\em { A Proof-theoretic
  Approach to Logic Programming: II. Programs as Definitions }}, Journal of
  Logic and Computation, 1 (1991), pp.~635--660.

\bibitem{harrop1960concerning}
{\sc R.~Harrop}, {\em { Concerning formulas of the types $A \to B \lor C$,
  $A\to(Ex) B (x)$ in intuitionistic formal systems }}, The Journal of Symbolic
  Logic, 25 (1960), pp.~27--32.

\bibitem{heyting1966intuitionism}
{\sc A.~Heyting}, {\em {Intuitionism: An Introduction}}, Cambridge University
  Press, 1989.

\bibitem{Hofmann1997}
{\sc M.~Hofmann}, {\em {Syntax and Semantics of Dependent Types}}, Semantics
  and Logics of Computation,  (1997), p.~79–130.

\bibitem{howard1980formulae}
{\sc W.~A. Howard}, {\em {The Formulae-as-Types Notion of Construction}}, To H.
  B. Curry: Essays on Combinatory Logic, Lambda Calculus and Formalism, 44
  (1980), pp.~479--490.

\bibitem{Jacobs}
{\sc B.~Jacobs}, {\em {Categorical Type Theory}}, PhD thesis, The University of
  Nijmegen, 1991.

\bibitem{kolmogorov}
{\sc A.~Kolmogorov}, {\em {Zur Deutung der Intuitionistischen Logik}},
  Mathematische Zeitschift, 35 (1932).

\bibitem{kurbis2019proof}
{\sc N.~K{\"u}rbis}, {\em {Proof and Falsity: A Logical Investigation}},
  Cambridge University Press, 2019.

\bibitem{lambek1980lambda}
{\sc J.~Lambek}, {\em {From $\lambda$-calculus to Cartesian Closed
  Categories}}, To H. B. Curry: Essays on Combinatory Logic, Lambda Calculus
  and Formalism,  (1980), pp.~375--402.

\bibitem{makinson2014inferential}
{\sc D.~Makinson}, {\em {On an Inferential Semantics for Classical Logic}},
  Logic Journal of IGPL, 22 (2014), pp.~147--154.

\bibitem{martin1975intuitionistic}
{\sc P.~Martin-L{\"o}f}, {\em {An Intuitionistic Theory of Types: Predicative
  Part}}, Studies in Logic and the Foundations of Mathematics, 80 (1975),
  pp.~73--118.

\bibitem{miller1989logical}
{\sc D.~Miller}, {\em {A Logical Analysis of Modules in Logic Programming}},
  The Journal of Logic Programming, 6 (1989), pp.~79--108.

\bibitem{nascimentothesis}
{\sc V.~Nascimento}, {\em {Foundational Studies in Proof-theoretic Semantics}},
  PhD thesis, Universidade do Estado do Rio de Janeiro, 2023.

\bibitem{Parigot2005}
{\sc M.~Parigot}, {\em { $\lambda\mu$-Calculus: An algorithmic interpretation
  of classical natural deduction }}, in Logic Programming and Automated
  Reasoning, A.~Voronkov, ed., Springer, 1992, pp.~190--201.

\bibitem{Pavlovic1990}
{\sc D.~Pavlovi\'c}, {\em {Predicates and Fibrations}}, PhD thesis, University
  of Utrecht, 1990.

\bibitem{Pereira2024}
{\sc L.~C. Pereira, E.~H. Haeusler, and V.~Nascimento}, {\em {Disjunctive
  Syllogism without Ex falso}}, in Peter Schroeder-Heister on Proof-Theoretic
  Semantics, T.~Piecha and K.~F. Wehmeier, eds., Springer, 2024, pp.~193--209.

\bibitem{Piecha2016completeness}
{\sc T.~Piecha}, {\em {Completeness in Proof-theoretic Semantics}}, in Advances
  in Proof-theoretic Semantics, Springer, 2016, pp.~231--251.

\bibitem{Piecha2015failure}
{\sc T.~Piecha, W.~de~Campos~Sanz, and P.~Schroeder-Heister}, {\em {Failure of
  Completeness in Proof-theoretic Semantics}}, Journal of Philosophical Logic,
  44 (2015), pp.~321--335.

\bibitem{Piecha2017definitional}
{\sc T.~Piecha and P.~Schroeder-Heister}, {\em {The Definitional View of Atomic
  Systems in Proof-theoretic Semantics}}, in The Logica Yearbook 2016, College
  Publications London, 2017, pp.~185--200.

\bibitem{Piecha2019incompleteness}
\leavevmode\vrule height 2pt depth -1.6pt width 23pt, {\em { Incompleteness of
  Intuitionistic Propositional Logic with Respect to Proof-theoretic Semantics
  }}, Studia Logica, 107 (2019), pp.~233--246.

\bibitem{Pogorzelski}
{\sc W.~A. Pogorzelski}, {\em {Structural completeness of the propositional
  calculus}}, Bulletin de l'Acad\'emie Polonaise des Sciences, 19 (1971),
  pp.~349--351.

\bibitem{Prawitz1971ideas}
{\sc D.~Prawitz}, {\em {Ideas and Results in Proof Theory}}, in Studies in
  Logic and the Foundations of Mathematics, vol.~63, Elsevier, 1971,
  pp.~235--307.

\bibitem{Prawitz1972}
\leavevmode\vrule height 2pt depth -1.6pt width 23pt, {\em {The Philosophical
  Position of Proof Theory}}, in Contemporary Philosophy in Scandinavia, R.~E.
  Olson and A.~M. Paul, eds., John Hopkins Press, 1972, pp.~123--134.

\bibitem{Prawitz1973towards}
\leavevmode\vrule height 2pt depth -1.6pt width 23pt, {\em {Towards a
  Foundation of a General Proof Theory}}, in Studies in Logic and the
  Foundations of Mathematics, vol.~74, Elsevier, 1973, pp.~225--250.

\bibitem{Prawitz1974}
\leavevmode\vrule height 2pt depth -1.6pt width 23pt, {\em {On the Idea of a
  General Proof Theory}}, Synthese, 27 (1974), pp.~63--77.

\bibitem{Prawitz2006natural}
\leavevmode\vrule height 2pt depth -1.6pt width 23pt, {\em {Natural Deduction:
  A Proof-theoretical Study}}, Dover Publications, 2006 [1965].

\bibitem{Prawitz2007}
{\sc D.~Prawitz}, {\em {Logical Consequence From a Constructivist View}}, in
  The Oxford Handbook of Philosophy of Mathematics and Logic, Oxford University
  Press, 2007.

\bibitem{PR2001disj}
{\sc D.~Pym and E.~Ritter}, {\em {On the semantics of classical disjunction}},
  Journal of Pure and Applied Algebra, 159 (2001), pp.~315--338.

\bibitem{pym2004reductive}
{\sc D.~J. Pym and E.~Ritter}, {\em {Reductive Logic and Proof-search: Proof
  Theory, Semantics, and Control}}, Oxford University Press, 2004.

\bibitem{Pym2022catpts}
{\sc D.~J. Pym, E.~Ritter, and E.~Robinson}, {\em {Proof-theoretic Semantics in
  Sheaves (Extended Abstract)}}, in Proceedings of the Eleventh Scandinavian
  Logic Symposium --- SLSS 11, 2022.

\bibitem{pym2024categorical}
\leavevmode\vrule height 2pt depth -1.6pt width 23pt, {\em Categorical
  proof-theoretic semantics}, Studia Logica,  (2024), pp.~1--38.

\bibitem{reiter1981closed}
{\sc R.~Reiter}, {\em {On closed world data bases}}, in Readings in Artificial
  Intelligence, Elsevier, 1981, pp.~119--140.

\bibitem{Takemura}
{\sc T.~Ryo}, {\em { Investigation of Prawitz's completeness conjecture in
  phase semantic framework }}, Journal of Humanities and Sciences, 23 (2017).

\bibitem{Sandqvist2009CL}
{\sc T.~Sandqvist}, {\em {Classical Logic without Bivalence}}, Analysis, 69
  (2009), pp.~211--218.

\bibitem{Sandqvist2015IL}
\leavevmode\vrule height 2pt depth -1.6pt width 23pt, {\em {Base-extension
  Semantics for Intuitionistic Sentential Logic}}, Logic Journal of the IGPL,
  23 (2015), pp.~719--731.

\bibitem{sandqvist2015hypothesis}
\leavevmode\vrule height 2pt depth -1.6pt width 23pt, {\em
  {Hypothesis-discharging Rules in Atomic Bases}}, in Dag Prawitz on Proofs and
  Meaning, Springer, 2015, pp.~313--328.

\bibitem{sandqvistwld}
\leavevmode\vrule height 2pt depth -1.6pt width 23pt, {\em {Atomic bases and
  the validity of Peirce's law}}.
\newblock
  \url{https://sites.google.com/view/wdl-ucl2022/schedule#h.ttn75i73elfw},
  2022.
\newblock {World Logic Day --- University College London (Accessed June 2023)}.

\bibitem{schroeder1984natural}
{\sc P.~Schroeder-Heister}, {\em {A natural extension of natural deduction}},
  The Journal of Symbolic Logic, 49 (1984), pp.~1284--1300.

\bibitem{schroeder1993rules}
\leavevmode\vrule height 2pt depth -1.6pt width 23pt, {\em {Rules of
  Definitional Reflection}}, in Logic in Computer Science --- LICS, IEEE, 1993,
  pp.~222--232.

\bibitem{Schroeder2006validity}
\leavevmode\vrule height 2pt depth -1.6pt width 23pt, {\em {Validity Concepts
  in Proof-theoretic Semantics}}, Synthese, 148 (2006), pp.~525--571.

\bibitem{Schroeder2007modelvsproof}
\leavevmode\vrule height 2pt depth -1.6pt width 23pt, {\em {Proof-Theoretic
  versus Model-Theoretic Consequence}}, in The Logica Yearbook 2007, M.~Pelis,
  ed., Filosofia, 2008.

\bibitem{Schroeder2012categorical}
\leavevmode\vrule height 2pt depth -1.6pt width 23pt, {\em { The Categorical
  and the Hypothetical: A Critique of Some Fundamental Assumptions of Standard
  Semantics }}, Synthese, 187 (2012), pp.~925--942.

\bibitem{Schroeder2015elim}
\leavevmode\vrule height 2pt depth -1.6pt width 23pt, {\em {Proof-theoretic
  Validity based on Elimination Rules}}, in Why is this a Proof? Festschrift
  for Luiz Carlos Pereira, College Publications, 2015.

\bibitem{SEP-PtS}
\leavevmode\vrule height 2pt depth -1.6pt width 23pt, {\em {Proof-Theoretic
  Semantics}}, in The Stanford Encyclopedia of Philosophy, E.~N. Zalta, ed.,
  Metaphysics Research Lab, Stanford University, {S}pring 2018~ed., 2018.

\bibitem{Schroeder2016atomic}
{\sc P.~Schroeder{-}Heister and T.~Piecha}, {\em {Atomic Systems in
  Proof-Theoretic Semantics: Two Approaches}}, in Epistemology, Knowledge and
  the Impact of Interaction, \'{A}ngel Nepomuceno~Fern\'{a}ndez, O.~P. Martins,
  and J.~Redmond, eds., Springer Verlag, 2016.

\bibitem{seely1983hyperdoctrines}
{\sc R.~A. Seely}, {\em {Hyperdoctrines, Natural Deduction and the Beck
  condition}}, Mathematical Logic Quarterly, 29 (1983), pp.~505--542.

\bibitem{Stafford2021}
{\sc W.~Stafford}, {\em {Proof-Theoretic Semantics and Inquisitive Logic}},
  Journal of Philosophical Logic,  (2021).

\bibitem{stafford2023}
{\sc W.~Stafford and V.~Nascimento}, {\em { Following all the Rules:
  Intuitionistic Completeness for Generalized Proof-theoretic Validity }},
  Analysis,  (2023).

\bibitem{Streicher1988}
{\sc T.~Streicher}, {\em { Correctness and Completeness of a Categorical
  Semantics of the Calculus of Constructions }}, PhD thesis, University of
  Passau, 1988.

\bibitem{Gentzen}
{\sc M.~E. Szabo}, ed., {\em {The Collected Papers of Gerhard Gentzen}},
  North-Holland Publishing Company, 1969.

\bibitem{Tennant78entailment}
{\sc N.~Tennant}, {\em {Entailment and Proofs}}, Proceedings of the
  Aristotelian Society, 79 (1978), pp.~167--189.

\bibitem{Tennant2017core}
\leavevmode\vrule height 2pt depth -1.6pt width 23pt, {\em {Core Logic}},
  Oxford University Press, 2017.

\bibitem{wansing2000idea}
{\sc H.~Wansing}, {\em { The idea of a proof-theoretic semantics and the
  meaning of the logical operations }}, Studia Logica, 64 (2000), pp.~3--20.

\end{thebibliography}
